\documentclass[a4paper,USenglish]{lipics-v2021}
\nolinenumbers

\newcommand{\ShortVersion}[1]{}

\newcommand{\LongVersion}[1]{#1}

\usepackage{xspace}
\usepackage{tikz}
\usepackage{tikz-qtree}
\usepackage{mathtools}
\usepackage{algorithm}
\usepackage[noend]{algorithmic}

\newcommand{\E}[1]{\mathop{\mathbb{E}\left[ #1 \right]}}
\newcommand{\Prob}[1]{\mathop{\mathrm{Pr}\left[ #1 \right]}}
\newcommand{\ceil}[1]{\lceil #1\rceil}
\newcommand{\etal}{\textit{et~al}.\xspace}
\newtheorem{fact}{Fact}
\newcommand{\proc}[1]{\fontfamily{cmr}\textup{\textsc{#1}}\xspace}
\newcommand{\RebalanceZig}{\proc{Rebalance\-Zig}}
\newcommand{\RebalanceZigZag}{\proc{Rebalance\-ZigZag}}
\newcommand{\RebalanceZigZig}{\proc{Rebalance\-ZigZig}}
\newcommand{\RotateUp}[1]{\proc{RotateUp}(#1)}
\newcommand{\nil}{\textup{\textsc{nil}}\xspace}
\let\epsilon\varepsilon

\graphicspath{{plots/}}

\newcommand{\plot}[2][]{\resizebox{0.49\textwidth}{!}{\includegraphics{#2}}\llap{\raisebox{0.34\textwidth}{\makebox[0.48\textwidth][l]{\hspace{1ex}\scriptsize\textbf{#1}}}}}

\ShortVersion{
  \EventEditors{Andrei Z. Broder and Tami Tamir}
  \EventNoEds{2}
  \EventLongTitle{12th International Conference on Fun with Algorithms (FUN 2024)}
  \EventShortTitle{FUN 2024}
  \EventAcronym{FUN}
  \EventYear{2024}
  \EventDate{June 4--8, 2024}
  \EventLocation{Island of La Maddalena, Sardinia, Italy}
  \EventLogo{}
  \SeriesVolume{291}
  \ArticleNo{25}
}

\LongVersion{
  \hideLIPIcs
}

\tikzset{
  missing/.style={draw=none},
  sibling distance=6pt, 
  level distance=5ex,
  every tree node/.style={draw, circle, minimum width=1.5em, inner sep=0.1ex},
  edge from parent/.style={draw,edge from parent path={(\tikzparentnode) -- (\tikzchildnode)}},
}

\newcommand{\tree}[1]{\begin{tikzpicture}[baseline=(current bounding box.center)]\Tree #1\end{tikzpicture}}
\newcommand{\ARROW}[2]{$\xrightarrow[\clap{\text{\scriptsize #2}}]{\clap{\text{\scriptsize #1}}}$}

\title{Bottom-up Rebalancing Binary Search Trees by Flipping a Coin}

\author
  {Gerth Stølting Brodal}
  {Department of Computer Science, Aarhus University, Aabogade 34, 8200 Aarhus N, Denmark}
  {gerth@cs.au.dk}
  {0000-0001-9054-915X}
  {Supported by Independent Research Fund Denmark, grant~9131-00113B.}

\authorrunning{G.\,S. Brodal}

\Copyright{Gerth Stølting Brodal}



\ShortVersion{\relatedversiondetails[cite=fun24arxiv]{Full paper}{https://doi.org/10.48550/arXiv.2404.xxxxx}}

\LongVersion{\relatedversiondetails[cite=fun24]{FUN 2024 proceedings}{https://doi.org/10.4230/LIPIcs.FUN.2024.25}}

\begin{document}

\maketitle


\begin{abstract}
  Rebalancing schemes for dynamic binary search trees are numerous in the literature, where the goal is to maintain trees of low height, either in the worst-case or expected sense. In this paper we study randomized rebalancing schemes for sequences of $n$ insertions into an initially empty binary search tree, under the assumption that a tree only stores the elements and the tree structure without any additional balance information. Seidel~(2009) presented a top-down randomized insertion algorithm, where insertions take expected $O\big(\lg^2 n\big)$ time, and the resulting trees have the same distribution as inserting a uniform random permutation into a binary search tree without rebalancing. Seidel states as an open problem if a similar result can be achieved with bottom-up insertions. In this paper we fail to answer this question.
  
  We consider two simple canonical randomized bottom-up insertion algorithms on binary search trees, assuming that an insertion is given the position where to insert the next element. The subsequent rebalancing is performed bottom-up in expected $O(1)$ time, uses expected $O(1)$ random bits, performs at most two rotations, and the rotations appear with geometrically decreasing probability in the distance from the leaf. For some insertion sequences the expected depth of each node is proved to be $O(\lg n)$. On the negative side, we prove for both algorithms that there exist simple insertion sequences where the expected depth is $\Omega(n)$, i.e., the studied rebalancing schemes are \emph{not} competitive with (most) other rebalancing schemes in the literature.
\end{abstract}

\section{Introduction}
\label{sec:introduction}

Binary search trees is one of the most fundamental data structures in computer science, dating back to the early 1960s, see, e.g., Windley (1960)~\cite{Windley60} for an early description of binary search trees and Hibbard (1962)~\cite{Hibbard62} for an analysis of random insertions and deletions. Knuth~\cite[page 453]{Knuth98} gives a detailed  overview of the early history of binary search trees, and Andersson \etal~\cite{AnderssonFagerbergLarsen04} an overview of later developments on balanced binary search trees.

When inserting new elements into the leaves of an unbalanced binary search tree the height of the tree might deteriorate, in the sense that it becomes super-logarithmic in the number of elements stored (see Figure~\ref{fig:unbalanced-insertions}). In the literature numerous rebalancing schemes have been presented guaranteeing logarithmic height: Some are deterministic with worst-case update bounds, like AVL-trees~\cite{AVL62}, red-black trees~\cite{GuibasSedgewick78}; some deterministic with amortized bounds, like splay-trees~\cite{SleatorTarjan85} and scapegoat trees~\cite{Andersson89,GalperinRivest93}; and others are randomized, like treaps~\cite{SeidelAragon96} and randomized binary search trees~\cite{MartinezRoura98}, just to mention a few.

In this paper we study simple randomized rebalancing schemes for sequences of insertions into an initially empty binary search tree. The goal of this paper is to study randomized rebalancing schemes under a set of constraints, and to study how good rebalancing schemes can be achieved within these constraints. In general the proposed rebalancing schemes in Section~\ref{sec:zig} and Section~\ref{sec:zigzag} are \emph{not} competitive with existing rebalancing schemes in the literature. The constraints we consider are the following:

\begin{enumerate}
  \item\label{req:no-balance} The search tree should not store any balancing information, only the tree and the elements should be stored.
  \item\label{req:rotations} Insertions should perform limited restructuring, say, worst-case $O(1)$ rotations.
  \item\label{req:local} Most rotations should happen near the inserted elements.
  \item\label{req:no-information} Rebalancing should  be based on local information (tree structure) at the insertion point only (e.g., without knowledge of $n$ nor the current height of the tree).
  \item\label{req:time} Rebalancing should be performed in expected $O(1)$ time.
  \item\label{req:random-bits} Rebalancing should use expected few random bits per insertion, say, expected $O(1)$ bits.
  \item\label{req:depth} Each node should have low expected depth, ideally $O(\lg n)$. 
\end{enumerate}

The constraints are motivated by the properties of the randomized treaps~\cite{SeidelAragon96} (but treaps need to store random priorities as balance information); that the random distribution of tree structures achieved by treaps can be achieved without storing balancing information~\cite{Seidel09} (but with slower insertions); and treaps can be the basis for efficient concurrent search trees~\cite{AlapatiSuranamMutyam17}.

\begin{figure}
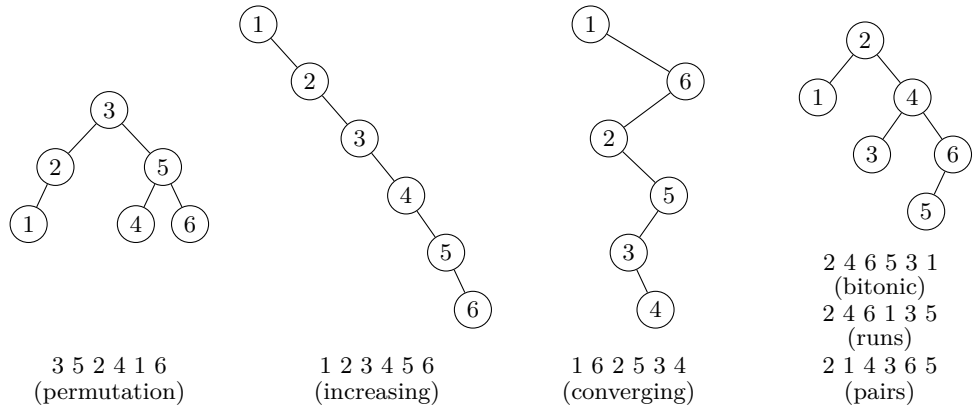

  \centering
  \tabcolsep10pt
  \begin{tabular}{cccc}
    \tree{[.3 [.2 [.1 ] \edge[missing]; \node[missing] {}; ] [.5 [.4 ] [.6 ] ] ]} &
    \hspace{-3em}\tree{[.1 \edge[missing]; \node[missing] {}; [.2 \edge[missing]; \node[missing] {}; [.3 \edge[missing]; \node[missing] {}; [.4 \edge[missing]; \node[missing] {}; [.5 \edge[missing]; \node[missing] {}; [.6 ] ] ] ] ] ]} &
    \makebox[2.2cm]{\tree{[.1 \edge[missing]; \node[missing] {}; [.6 [.2 \edge[missing]; \node[missing] {}; [.5 [.3 \edge[missing]; \node[missing] {}; 4 ] \edge[missing]; \node[missing] {}; ] ] \edge[missing]; \node[missing] {}; ] ]}} &
    \makebox[3cm]{\hspace{1.5em}\raisebox{4ex}{\tree{[.2 1 [.4 3 [.6 5 \edge[missing]; \node[missing] {}; ] ] ]}}} \\[15ex]
    3 5 2 4 1 6 & 1 2 3 4 5 6 & 1 6 2 5 3 4 & 
    \smash{\begin{minipage}[b]{6em}
    \centering
    2 4 6 5 3 1 \\[-0.7ex]
    (bitonic) \\[-0.25ex]
    2 4 6 1 3 5 \\[-0.7ex]
    (runs) \\[-0.25ex]
    2 1 4 3 6 5
    \end{minipage}}
    \\[-0.7ex]
    (permutation) & (increasing) & (converging) & (pairs)
  \end{tabular}
  \caption{Unbalanced binary search trees resulting from inserting permutations of $\{1,\ldots,6\}$. The insertion order is shown below the trees. The types of permutations are defined in Table~\ref{tab:sequences}.}
  \label{fig:unbalanced-insertions}
\end{figure}

\subsection{Deterministic Previous Work}

Red-black trees~\cite{GuibasSedgewick78} are deterministic dynamic balanced binary search, with good amortized performance. They violate constraint~(\ref{req:no-balance}), since each node is required to store a single bit of balance information, indicating if the node is red-black. But otherwise, red-black trees are guaranteed to have height $O(\lg n)$, insertions at a leaf can be performed in amortized~$O(1)$ time and perform at most two rotations, i.e., red-black trees essentially satisfy constraints~(\ref{req:rotations})--(\ref{req:depth}), if expected bounds are substituted by amortized bounds.

Brown~\cite{Brown78,Brown79} showed how to encode a single bit of information in the internal nodes of a binary tree by considering ``supernodes'' consisting of pairs of consecutive elements arranged as parent-child pairs together with a pointer to an empty leaf between the two elements. Depending of the relative placement of the two elements the encoded bit can be decoded from the placement of the pointer to the empty leaf. Brown showed how to encode 2-3 trees~\cite[Chapter~4]{AhoHopcroftUllman74} using this technique, achieving balanced binary search trees storing no balance information and supporting insertions in worst-case $O(\lg n)$ time.

Splay trees~\cite{SleatorTarjan85} are the canonical deterministic amortized efficient dynamic binary search trees satisfying constraint~(\ref{req:no-balance}), i.e., they do not store any additional information than the binary tree and the elements. Splay trees support insertions in amortized~$O(\lg n)$ time, i.e., insertions are amortized efficient. The drawback of splay trees is that they do a significant amount of restructuring (memory updates) per insertion, since they rotate a constant fraction of the nodes on the path from the inserted element to the root. The number of rotations depends on what variant of splaying is applied, see~\cite{BrinkmannDegraerLoof09,SleatorTarjan85}. Albers and Karpinski~\cite{AlbersKarpinski02} and F\"urer~\cite{Furer99} considered randomized variants of splaying to reduce the restructuring cost. 

Scapegoat trees are another deterministic dynamic binary search tree with good amortized performance, independently discovered by Anderson~\cite{Andersson89,Andersson99} and Galperin and Rivest~\cite{GalperinRivest93}. Scapegoat trees achieve amortized $O(\lg n)$ insertions by maintaining the invariant that the height of a tree containing $n$ elements is $O(\lg n)$. If an insertion causes the invariant to be violated, a local subtree is rebuild into a perfectly balanced binary tree (in the worst-case this is the root and the full tree is rebuild). A scapegoat tree only needs to store the number of elements~$n$ as a global integer value in addition to the binary tree and its elements. 

\subsection{Randomized Previous Work}

Randomization and binary search trees can be addressed in two directions in the context of insertions: Either insertions are random (e.g., the insertion sequence is a random permutation) and we analyze the expected performance for a binary search tree with respect to the insertion distribution; or insertions can be arbitrary but the rebalancing of the binary search tree exploits random bits and we analyze the performance with respect to the random bits.

We call a search tree containing $n$ elements inserted in random order without rebalancing a \emph{random binary search tree}. A classic result on random binary search trees is that each element in the resulting tree has expected depth at most $2\ln n + O(1)$~\cite{BoothCollin60,Hibbard62,Windley60}. The important property is that the root equals each of the $n$ elements with probability exactly~$1/n$, and this property again holds recursively for the left and right subtrees. A consequence is that all valid search trees with $n \geq 3$ elements do not have the same  probability. 
\LongVersion{See Figure~\ref{fig:unbalanced-distribution-3-nodes} for the minimal case~$n=3$. }%
See Panny~\cite{Panny10} for a history on deletions in random binary search trees.

\LongVersion{%
\begin{figure}
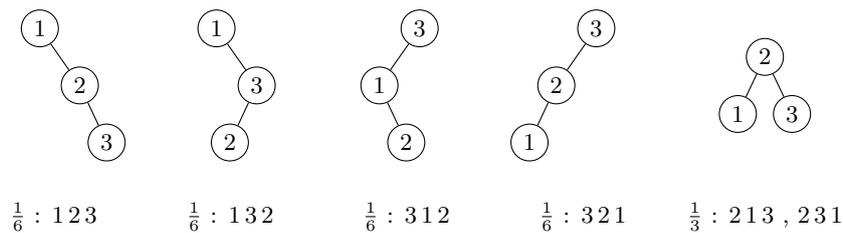

  \centering
  \begin{tabular}{cccccc}
    \tree{[.1 \edge[missing]; \node[missing] {}; [.2 \edge[missing]; \node[missing] {}; 3 ] ]} &
    \tree{[.1 \edge[missing]; \node[missing] {}; [.3 2 \edge[missing]; \node[missing] {}; ] ]} &
    \tree{[.3 [.1 \edge[missing]; \node[missing] {}; 2 ] \edge[missing]; \node[missing] {}; ]} &
    \tree{[.3 [.2 1 \edge[missing]; \node[missing] {}; ] \edge[missing]; \node[missing] {}; ]} &
    \tree{[.2 1 3 ]} \\[9ex]
    $\frac{1}{6}$ : 1\,2\,3 & 
    $\frac{1}{6}$ : 1\,3\,2 & 
    $\frac{1}{6}$ : 3\,1\,2 & 
    $\frac{1}{6}$ : 3\,2\,1 & 
    $\frac{1}{3}$ : 2\,1\,3 , 2\,3\,1
  \end{tabular}  
  \caption{The unbalanced binary search trees resulting from inserting (without rebalancing) random permutations of the sequence 1, 2, 3. Below each tree is the probability of the tree and the permutations resulting in the tree.}
  \label{fig:unbalanced-distribution-3-nodes}
\end{figure}
}

The structure of random binary search trees has been used as guideline to construct different dynamic binary search trees where at each point of time the probability of a given tree equals that of a random binary search tree~\cite{MartinezRoura98,Seidel09,SeidelAragon96}. 

Aragon and Seidel~\cite{SeidelAragon96} introduced the \emph{treap}, that with each element stores an independently uniformly assigned random priority in the range $[0,1]$, and organizes the search tree such that priorities satisfy heap order~\cite{Williams64}, i.e., the root stores the element with minimum priority. Since each element has probability $1/n$ to have the smallest priority, all elements have probability~$1/n$ to be at the root, the property required to be random search trees. Insertions into treaps can be done by bottom-up rotations in expected $O(1)$ time and $O(1)$ rotations using $O(1)$ random bits. After $n$ insertions into a treap the expected shape of the treap equals a random binary search tree. Blelloch and Reid-Miller~\cite{BlellochR98} considered parallel algorithms for the set operations union, intersection and difference on treaps. Alapati \etal~\cite{AlapatiSuranamMutyam17} considered concurrent insetions and deletions into treaps.

Mart{\'{\i}}nez and Roura~\cite{MartinezRoura98} presented a different approach denoted \emph{randomized binary search trees} to achieve the structure of a random binary search tree after $n$ insertions. Their approach stores at each node the size of the subtree rooted at the node, and insertions are performed top-down in expected $O(\lg n)$ time, where the inserted element is inserted in a node with probability $\frac{1}{k+1}$, where $k$ is the size of the current subtree rooted at the node (see~\cite{MartinezRoura98} for details). Each insertion requires expected $O(\lg n)$ random integers in the range~$1,\ldots,n+1$.

Seidel~\cite{Seidel09} gave a unified presentation of \cite{SeidelAragon96} and \cite{MartinezRoura98}, emphasizing the similarity of the two approaches, and describes a variation of \cite{MartinezRoura98} that avoids storing subtree sizes, but the insertion time increases to expected $O\big(\lg^2 n\big)$ and uses $O\big(\lg^3 n\big)$ random bits. Seidel states it as an open problem if there exists a bottom-up rebalancing algorithm that without storing any balancing information can obtain the structure of random binary search trees.

\begin{table}
  \caption{Rebalancing cost of selected binary search trees. Space refers to the space required for additional balance information. $O_A$, $O_E$ and $O$ denote amortized, expected and worst-case bounds, respectively. *Our results do not guarantee (expected) logarithmic depth.}
  \label{tab:results}
  \centering
  \begin{tabular}{lcccc}
    & Time & Rotations & Random bits & Space (bits) \\
    \hline
    Red-black tree~\cite{GuibasSedgewick78} & $O_A(1)$ & $O(1)$ & 0 & 1 \\
    Encoded 2-3 trees~\cite{Brown78,Brown79} & $O(\lg n)$ & $O(\lg n)$ & 0 & 0 \\
    Splay trees & $O_A(\lg n)$ & $O_A(\lg n)$ & 0 & 0 \\
    Treaps~\cite{AragonSeidel89} & $O_E(1)$ & $O_E(1)$ & $O_E(1)$ & $O_E(1)$ \\
    Randomized BST~\cite{MartinezRoura98} & $O_E(\lg n)$ & $O_E(1)$ & $O_E(\lg^2 n)$ & $O(\lg n)$ \\
    Seidel~\cite{Seidel09} & $O_E(\lg^2 n)$ & $O_E(1)$ & $O_E(\lg^3 n)$ & 0 \\
    \emph{Algorithms in this paper*} & $O_E(1)$ & $O(1)$ & $O_E(1)$ & 0 \\
    \hline    
  \end{tabular}
\end{table}

\subsection{Results}

\begin{table}
  \caption{Different sequences of length $n$ (assuming $n$ is even) considered in this paper.}
  \label{tab:sequences}
  \centering
  \begin{tabular}{ll}
    \hline
    \textbf{permutation} & random permutation of $1,\ldots,n$ \\
    \textbf{increasing} & $1,2,3,\ldots,n$ \\
    \textbf{decreasing} & $n,n-1,n-2,\ldots,1$ \\
    \textbf{converging} & $1,n,2,n-1,3,n-2,\ldots,\frac{n}{2},\frac{n}{2}+1$ \\
    \textbf{pairs} & $2,1,4,3,6,5,\ldots,n,n-1$ \\
    \textbf{bitonic} & $2,4,6,\ldots,n-2,n,n-1,n-3,\ldots,5,3,1$ \\
    \textbf{runs} & $2,4,6,\ldots,n-2,n,1,3,5,\ldots,n-3,n-1$ \\
    \hline
  \end{tabular}
\end{table}

We consider two very simple algorithms to rebalance a binary search tree after a new element has been inserted at a leaf. Our aim is to try to meet the requirements (\ref{req:no-balance})--(\ref{req:depth}), and in particular not the ambitious goal of having the same distribution as random binary search trees. Both our algorithms repeatedly flip a coin until it comes out head. Whenever the coin shows tail (with probability $p$) we move to the parent of the current node (starting at the new leaf, and if we reach the root, the rebalancing terminates without modifying the tree). When the coin shows head, the first algorithm (\RebalanceZig in Algorithm~\ref{alg:RebalanceZig}) rotates the current node up, and the rebalancing terminates. The second algorithm (\RebalanceZigZag in Algorithm~\ref{alg:RebalanceZigZag}) does one or two rotations, depending on if it is a zig-zag or zig-zig case (inspired by the rebalancing rules of splay trees).

Ignoring the depths of the nodes of the resulting trees, we immediately have the following fact, since the coin tosses are independent Bernoulli trials, with an expected $O(1)$ coin tosses necessary (assuming a coin with constant non-zero probability for head). It follows that both algorithms satisfy our constraints (\ref{req:no-balance})--(\ref{req:random-bits}).

\begin{fact}
  The rebalancing done by \RebalanceZig with $0\leq p <1$ takes expected $O(1)$ time, uses expected $O(1)$ random bits, and performs at most one rotation. \RebalanceZigZag performs at most two rotations, but otherwise with identical performance.
\end{fact}

To study to what extend the proposed algorithms achieve logarithmic depth of the nodes, constraint~(\ref{req:depth}), we study the behavior of the algorithms on the insertion sequences listed in Table~\ref{tab:sequences}.

In Section~\ref{sec:zig} we study \RebalanceZig. Our first result is that \RebalanceZig is sufficient to achieve a balanced binary search tree when inserting elements in increasing order (and symmetrically decreasing order). The below theorem follows from Lemma~\ref{lem:zig-increasing}.

\begin{theorem}
  \label{thm:zig-increasing}
  Executing \RebalanceZig with $0<p<1$ on an increasing and decreasing sequence of $n$ insertions results in a binary search tree, where each node has expected depth~$O(\lg n)$.
\end{theorem}

We then show that there are very simple insertion sequences where \RebalanceZig fails to achieve a balanced tree. We denote the sequence $1, n, 2, n-1,\ldots, n/2, n/2+1$ the \emph{converging} sequence.  The below theorem follows from Lemma~\ref{lem:zig-converging}.

\begin{theorem}
  \label{thm:zig-converging}
  Executing \RebalanceZig with $0\leq p \leq 1$ on a converging sequence of $n$ insertions results in a binary search tree with expected average node depth~$\Theta(n)$.
\end{theorem}

Another sequence where \RebalanceZig fails to achieve logarithmic depth is on the \emph{pairs} sequence $2,1,4,3,6,5,\ldots,n,n-1$, provided $p\neq\frac{1}{2}$. The following theorem restates Lemma~\ref{lem:zig-pairs}.

\begin{theorem}
  \label{thm:zig-pairs}
  Executing \RebalanceZig with $0\leq p < \frac{1}{2}$ or $\frac{1}{2}<p\leq 1$ on a pairs sequence of $n$ insertions results in a binary search tree with expected average node depth~$\Theta(n)$.
\end{theorem}

For $p=\frac{1}{2}$ algorithm \RebalanceZig behaves significantly better on pairs sequences%
\LongVersion{, see the experimental evaluation in Figure~\ref{fig:zig-pairs}}%
. We conjecture the expected average node depth to be $O\big(\sqrt{n}\big)$%
\LongVersion{, motivated by Theorem~\ref{thm:process} in the appendix, studying a random process}%
.

In Section~\ref{sec:zigzag} we study the second algorithm \RebalanceZigZag. For increasing (decreasing) sequences, where the new leaf is always the rightmost (leftmost) node in the tree, \RebalanceZigZag is essentially identical to \RebalanceZig, i.e., our result for increasing and decreasing sequences for \RebalanceZig immediately carries over to \RebalanceZigZag. The following theorem restates Corollary~\ref{cor:zigzag-increasing}.
 
\begin{theorem}
  \label{thm:zigzag-increasing}
  Executing \RebalanceZigZag with $0<p<1$ on an increasing or decreasing sequence of $n$ insertions results in a binary search tree where each nodes has expected depth~$O(\lg n)$.
\end{theorem}

In Section~\ref{sec:zigzag-finger} we generalize the proof to also hold for the convergent sequence for \RebalanceZigZag (where \RebalanceZig failed to achieve logarithmic depth), and more generally finger sequences, where the next insertion always becomes the successor or predecessor of the last insertion. The following theorem restates Lemma~\ref{lem:zigzag-finger}.

\begin{theorem}
  \label{thm:zigzag-finger}
  Executing \RebalanceZigZag with $\frac{1}{2} \left(\sqrt{5} - 1\right)<p<1$ on a convergent or finger sequence of $n$ insertions results in a binary search tree where each nodes has expected depth~$O(\lg n)$.
\end{theorem}

On the negative side, we prove that \RebalanceZigZag fails to achieve balanced trees for pairs sequences, for all $0\leq p \leq 1$. The following theorem restates Lemma~\ref{lem:zigzag-pairs}.

\begin{theorem}
  \label{thm:zigzag-pairs}
  Executing \RebalanceZigZag with $0 \leq p \leq 1$ on a pairs sequence of $n$ insertions results in a binary search tree with expected average node depth~$\Theta(n)$.
\end{theorem}

We complement our theoretical findings by an experimental evaluation of \RebalanceZig and \RebalanceZigZag
\LongVersion{(and a third unpromising algorithm \RebalanceZigZig) }%
in Section~\ref{sec:experiments}, supporting our theoretical findings. We briefly discuss random permutations in Section~\ref{sec:random-permutation}, but otherwise only have an experimental evaluation of the rebalancing algorithms on inserting random permutations.

If the insights from our results can lead to an improved bottom-up randomized rebalancing scheme for binary search trees remains open.

\subsection{Notation and Terminology}

Throughout this paper $n$ denotes the number nodes in a binary search tree, i.e., the number of insertions performed. The \emph{depth} of a node is the number of edges from the node to the root, i.e., the root has depth zero. The height of a tree is the maximum depth of a node. Rebalancing will be done by the standard primitives of left and right \emph{rotations}, see Figure~\ref{fig:rotation}. Both rotate up a node one level in the tree. Since our updates are performed bottom-up, we assume that each node~$v$ stores an element and pointers to its left child~$v.l$, right child~$v.r$, and parent~$v.p$ (possibly equal to \nil if no such node exists). We let $\lg n$ and $\ln n$ denote the binary and natural logarithm of $n$, respectively.

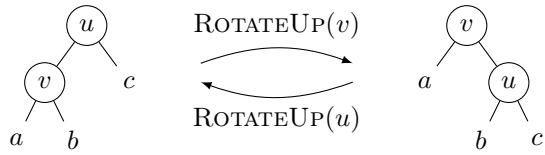
\begin{figure}[t]
  \centering
  \begin{tikzpicture}
    \begin{scope}[xshift=-2.5cm]
      \Tree [.{$u$} [.{$v$} \node[draw=none] {$a$}; \node[draw=none] {$b$}; ] \node[draw=none] {$c$}; ]
    \end{scope}
    \draw[-latex] (-1, -0.5) to[bend left=20] node[above, midway] {$\RotateUp{v}$} (1, -0.5);
    \draw[-latex] (1, -0.75) to[bend left=20] node[below, midway] {$\RotateUp{u}$} (-1, -0.75);
    \begin{scope}[xshift=2.5cm]
      \Tree [.{$v$} \node[draw=none] {$a$}; [.{$u$} \node[draw=none] {$b$}; \node[draw=none] {$c$}; ] ]
    \end{scope}
  \end{tikzpicture}
  \caption{(Left-to-right) The right rotation of $u$ rotates $v$ up; note that $a$ is moved one level up in the tree, $b$ remains at the same level, and $c$ is moved down one level. (Right-to-left) the left rotation of $v$ rotates $u$ up.}
  \label{fig:rotation}
\end{figure}

\section{Algorithm \RebalanceZig}
\label{sec:zig}

In this section we show that on increasing and decreasing sequences applying algorithm \RebalanceZig results in binary search trees where each node has expected depth $O(\lg n)$. We also show that on the converging and pairs sequences $\big(p\neq \frac{1}{2}\big)$ the expected average node depth is linear.

Assume a new element has been inserted into a binary search tree as a new leaf $v$ (before rebalancing the tree). Algorithm \RebalanceZig rebalances the tree as follows: After inserting the new node~$v$, we flip a coin, that with probability $p$ is tail and $1-p$ is head, for a constant $0 \leq p \leq 1$. If the coin is head, we rotate~$v$ up, and the insertion terminates. Otherwise, we recursively move to the parent, i.e., set $v \gets v.p$, flip a coin, and rotate the parent up if the coin is head, or continue recursively at the grandparent if the coin is tail. The rebalancing terminates when the first rotation has been performed or when we reach the root. See the pseudo-code in Algorithm~\ref{alg:RebalanceZig}. 

Note that $p=1$ is the special case where we always move up and never rotate, i.e., identical to insertions without rebalancing. When $p=0$ the new node is always the node rotated up. In this case the tree is a single path containing all $n$ nodes, since inserting a node~$v$ as a child of $u$ on the path, rotating up $v$ causes $v$ to be inserted into the path as the parent $u$. See Figure~\ref{fig:probability-zero}. In the following we assume $0<p<1$. That \RebalanceZig can not achieve the same tree distribution as random binary search trees (like treaps and randomize binary search trees do) follows by the example in Figure~\ref{fig:rotated-insertions}%
\LongVersion{ (compare with the result in Figure~\ref{fig:unbalanced-distribution-3-nodes})}%
.

\begin{figure}[t]
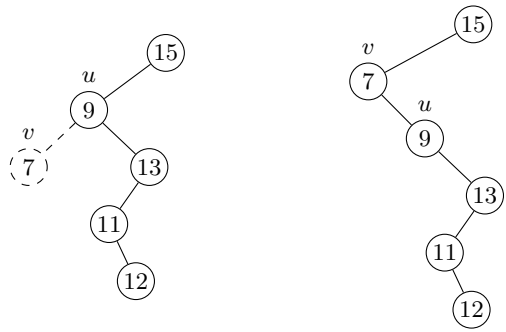

  \centering
  \begin{tabular}{cc}
    \tree{[.15 [.\node[label={$u$}]{9}; \edge[dashed]; \node[dashed, label={$v$}]{7}; [.13 [.11 \edge[missing]; \node[missing] {}; 12 ] \edge[missing]; \node[missing] {}; ] ] \edge[missing]; \node[missing] {}; ]} &
    \tree{[.15 [.\node[label={$v$}]{7}; \edge[missing]; \node[missing] {}; [.\node[label={$u$}]{9}; \edge[missing]; \node[missing] {}; [.13 [.11 \edge[missing]; \node[missing] {}; 12 ] \edge[missing]; \node[missing] {}; ] ] ] \edge[missing]; \node[missing] {}; ]}
  \end{tabular}
  \caption{\RebalanceZig with $p=0$ always rotates up the inserted node, and maintains the invariant that the tree is a single path. (left) insertion point of 7; (right) 7 is rotated up onto the path.}
  \label{fig:probability-zero}
\end{figure}

\begin{figure}[t]
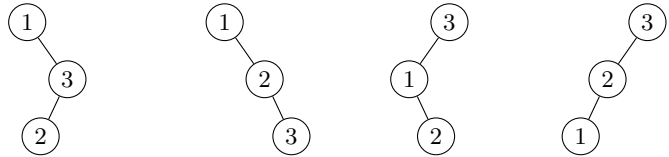

  \centering
  \tabcolsep10pt
  \begin{tabular}{ccccc}
    \tree{[.1 \edge[missing]; \node[missing] {}; [.3 [.2 ] \edge[missing]; \node[missing] {}; ] ]} &
    \tree{[.1 \edge[missing]; \node[missing] {}; [.2 \edge[missing]; \node[missing] {}; [.3 ] ] ]} &
    \tree{[.3 [.1 \edge[missing]; \node[missing] {}; [.2 ] ] \edge[missing]; \node[missing] {}; ]} &
    \tree{[.3 [.2 [.1 ] \edge[missing]; \node[missing] {}; ] \edge[missing]; \node[missing] {}; ]}
  \end{tabular}
  \caption{Binary search trees resulting from inserting the sequence 1, 3, 2 using \RebalanceZig. Each of the four search trees has probability~1/4 when $p=1/2$%
  \LongVersion{ (see Figure~\ref{fig:insertion-cases} for detailed cases)}%
  . Note that the perfect balanced binary search tree on three nodes cannot by achieved \RebalanceZig on this insertion sequence.}
  \label{fig:rotated-insertions}
\end{figure}

\begin{algorithm}[ht]
  \caption{\textsc{RebalanceZig}$(v)$}
  \label{alg:RebalanceZig}
  \begin{algorithmic}
    \WHILE {$v.p \neq \nil$ \AND coin flip is tail}
      \STATE $v \gets v.p$
    \ENDWHILE
    \IF {$v.p \neq \nil$}
      \STATE{$\RotateUp{v}$}
    \ENDIF
  \end{algorithmic}
\end{algorithm}

\subsection{Increasing Sequences}
\label{sec:zig-increasing}

We let the \emph{right height} of a tree denote the depth of the rightmost node in the tree. When inserting elements in increasing order the new element will always be inserted as a rightmost node in the tree, i.e., the right height increases by one before any rebalancing is performed. If \RebalanceZig performs a (left) rotation on the rightmost path, the right height is reduced by one again, and the right height does not change by the insertion.

\begin{lemma}
\label{lem:zig-right-height-increase}
  If the right height is $d$ before inserting a new rightmost node and applying \RebalanceZig, then afterwards the right height is $d$ or $d+1$ with probability $1-p^{d+1}$ or~$p^{d+1}$, respectively.
\end{lemma}

\begin{proof}
  The right height increases if and only if no rotation is performed, i.e., we reach the root because the coin~$d+1$ times in a row shows tail, which happens with probability~$p^{d+1}$.
\end{proof}

\begin{lemma}
\label{lem:zig-right-height}
  After inserting $n$ elements in increasing order using \RebalanceZig, the right height is at most $\ceil{(c+1)\cdot\log_{1/p} n}$ with probability $1-1/n^c$, for any constant $c > 0$.
\end{lemma}

\begin{proof}
  Assume that the right height at some point of time during the insertions is $d=c'\cdot\lg n$. By Lemma~\ref{lem:zig-right-height-increase}, the probability that the next insertion increases the right height is $p^{d+1}$. The probability of any of the at most $n$ remaining insertions increases the right height is at most $np^{d+1} = n p ^{1 + c' \lg n} \leq n p ^{c' \lg n} = n^{1+c'\lg p} \leq n^{-c}$ for $c' \geq -(c+1)/\lg p$. It follows that the right height after $n$ insertions is at most $\ceil{-(c+1)/\lg p \cdot \lg n}=\ceil{(c+1)\log_{1/p} n}$ with probability $1-1/n^c$.
\end{proof}

Lemma~\ref{lem:zig-right-height} gives a high probability guarantee on the expected depth of the nodes on the rightmost path. We now prove an expected depth for all nodes in the tree.

\begin{lemma}
  \label{lem:zig-increasing}
  After inserting $n$ elements in increasing order using \RebalanceZig, with $0<p<1$, each node has expected depth $O\big(1/p\cdot\log_{1/p} n\big)$.
\end{lemma}

\begin{proof}
  Consider an element inserted in a node~$v$ during the sequence of insertions. The element goes through the following five phases:
  \begin{enumerate}
    \item\label{case:v-not-created} The element is not yet inserted. The less than $n$ elements inserted before $v$ create a tree with right height $O\big(\log_{1/p} n\big)$ with high probability (Lemma~\ref{lem:zig-right-height}).
    \item\label{case:v-created} $v$ is created as the rightmost node with depth $O\big(\log_{1/p} n\big)$ with high probability.
    \item\label{case:v-rightmost-path} $v$ remains on the rightmost path for the subsequent insertions until $v$'s right child $u$ becomes the target for being rotated up. An insertion that rotates up $v$ or an ancestor of~$v$ will decrease the depth of $v$. Rotations below $u$ do not change the depth of $v$.
    \item\label{case:v-rotated-out} $v$ is moved out of the rightmost path by rotating up the right child $u$ of $v$, making $v$ the left child of $u$. This increases  the depth of $v$ by one.
    \item\label{case:v-left-subtree} $v$ is in the left subtree of a node $u$ on the rightmost path ($u$ can change over the subsequent insertions, but the depth of the branching node $u$ can never increase).  Each insertion can affect the position of $v$ by the rotation performed: 
    \begin{enumerate}
      \item\label{case:v-path-no-rotation} No rotation is performed and the path from the root through $u$ to $v$ is unchanged. The right height increases by one.
      \item\label{case:v-path-ancestor-rotated-up} An ancestor of $u$ is rotated up, where $u$ remains the branching node to $v$, the depths of both $u$ and $v$ decrease by one.
      \item\label{case:v-path-branch-rotated-up} $u$ is rotated up, where $u$ remains the branching node to $v$, the depth of $u$ decreases by one, and the depth of $v$ stays unchanged.
      \item\label{case:v-increase-depth} Rotating up the right child~$w$ of $u$ increases the depth of $v$ by one and $w$ replaces $u$ as the branching node to $v$ on the rightmost path (with the same depth as $u$ had before the rotation).
      \item\label{case:v-path-unchanged} Rotations below the right child of $u$ do not change the path from the root through $u$ and $v$. 
    \end{enumerate}
  \end{enumerate}
  From Lemma~\ref{lem:zig-right-height} it follows that the depth of $v$ after phases~1--4 is $O\big(\log_{1/p} n\big)$, with high probability.
  Cases \ref{case:v-path-no-rotation}, \ref{case:v-path-ancestor-rotated-up}, \ref{case:v-path-branch-rotated-up} and \ref{case:v-path-unchanged} do not increase the depth of $v$. 
  What remains is to bound the expected number of times case~\ref{case:v-increase-depth} occurs and increases the depth of $v$ by one.
  For case~\ref{case:v-increase-depth} to happen, a coin must have been flipped at $w$ showing head. 
  Over all insertions in phase~\ref{case:v-left-subtree}, a subsequence of the insertions flips a coin at the child $w$ of the current branching node~$u$. 
  If an insertion flips a coin at $w$, there are two cases: 
  The coin shows head with probability $1-p$ and case \ref{case:v-increase-depth} happens; 
  or the coin shows tail with probability $p$, and case~\ref{case:v-path-no-rotation},  \ref{case:v-path-ancestor-rotated-up} or \ref{case:v-path-branch-rotated-up} happens.
  Since cases~\ref{case:v-path-no-rotation}, \ref{case:v-path-ancestor-rotated-up} and \ref{case:v-path-branch-rotated-up} at most happens $O\big(\log_{1/p} n\big)$ times with high probability
  (case~\ref{case:v-path-no-rotation} increases the right height; cases~\ref{case:v-path-ancestor-rotated-up} and~\ref{case:v-path-branch-rotated-up} decrease the depth of the branching node $u$ to $v$), 
  i.e., the coin shows tail at $w$ at most $O\big(\log_{1/p} n\big)$ times with high probability.
  Since the expected number of times we need to flip a coin to get a tail is $1/p$, the expected number of times we flip a coin at the right child $w$ of the branching node $u$ to $v$ is $O\big(1/p\cdot\log_{1/p} n\big)$, with high probability.
  This is then also an upper bound on the expected number of times the depth of $v$ can increase by case \ref{case:v-increase-depth}.
  It follows that with high probability, the expected depth of $v$ is $O\big(\log_{1/p} n + 1/p\cdot\log_{1/p} n\big)=O\big(1/p\cdot\log_{1/p} n\big)$.
  Since the depth of~$v$ is at most~$n-1$, the expected depth of $v$ is $O\big(1/p\cdot\log_{1/p} n\big)$ after all insertions is (without the high probability assumption).
\end{proof}

The following lemma states that the node rotated up is expected to be close to the inserted leaf, and states the number of coin flips as a function of the tail probability~$p$.

\begin{lemma}
\label{lem:zig-distance-rotation}
  The distance from the inserted node to the node rotated up by \RebalanceZig is expected at most $\frac{p}{1-p}$. The number of coin flips is at most $\frac{1}{1-p}$.
\end{lemma}

\begin{proof}
  The expected distance to the node rotated up is at most
  \[
    \sum_{d=0}^{\infty} d (1-p)p^{d}=(1-p)\sum_{d=0}^{\infty} dp^{d}=(1-p)\frac{p}{(1-p)^2}=\frac{p}{1-p}\;,
  \]
  since the new node inserted (at distance 0) is rotated up with probability $1-p$, its parent with probability $p(1-p)$, etc. The $d$'th ancestor is rotated up with probability $(1-p)p^d$, provided a $d$'th ancestor exists. The number of coin flips is one plus the distance, i.e., at most $1+\frac{p}{1-p}=\frac{1}{1-p}$.
\end{proof}

Note that inserting $n$ elements in decreasing order is the symmetric case to the increasing order where the new node is inserted as the leftmost node, and at most one right rotation is performed on the leftmost path. If follows that Lemmas~\ref{lem:zig-right-height} and \ref{lem:zig-distance-rotation} also apply to decreasing sequences, by replacing the rightmost path by the leftmost path in the arguments,

\subsection{Converging Sequences}
\label{sec:zig-converging}

Assume we have a \emph{finger} into the sorted inserted sequence of elements pointing to the most recently inserted element, and whenever a new element is inserted it must be the new predecessor or successor of the element at the finger, i.e., we can only insert elements that are in the interval defined by the current predecessor and successor of the element at the finger. We call sequences satisfying this property for \emph{finger insertions}. Increasing and decreasing sequences are examples of finger insertions.  

The following sequence of insertions also consists of finger insertions (for simplicity, we assume $n$ is even). We denote this insertion sequence the \emph{converging sequence}. See Figure~\ref{fig:unbalanced-insertions} for an illustration of $n=6$.
\[ 
  1, n, 2, n-1, 3, n-2, 4, n-3, \ldots, n/2, n/2+1 
\]
As can be seen in the experimental evaluation in Figure~\ref{fig:experiments-comparison}(a,b), the average depth appears to be linear for the nodes in a binary search tree resulting from applying \RebalanceZig to the converging sequence. The below lemma confirms this.

\begin{lemma}
  \label{lem:zig-converging}
  Executing \RebalanceZig with insertions $1, n, 2, n-1,\ldots, n/2, n/2+1$, assuming $n$ even, results in a binary search tree with an (external) leaf with expected depth at least~$\frac{p(1-p)}{2} n$, for $0<p<1$. For $p=0$ and $p=1$ the resulting tree is a single path.
\end{lemma}

\begin{proof}
  For $p=1$ we do no rebalancing, and the converging sequence results in a single path (see Figure~\ref{fig:unbalanced-insertions}). For $p=0$, the new node is always rotated up onto an existing single path. We let the \emph{insertion point} denote the (external) leaf, where the next insertion is going to create a node $v$. Since the converging sequence is a sequence of finger insertions, the next insertion point is always a child of the created node~$v$ (before rebalancing), i.e., each insertion increases the depth of insertion point by one (before rebalancing). Unfortunately, the rebalancing done by \RebalanceZig does not always decrease the depth one. Consider inserting $i$, where $1\leq i \leq n/2$, that creates a node $u$ followed by inserting $n+1-i$ that creates a node $v$. Assume that the rebalancing after inserting $i$ does not rotate up $u$ (but possibly an ancestor of $u$ has been rotated up, and possibly decreasing the depth of the insertion point by one again), and the insertion of $n+1-i$ causes $v$ to be rotated up. This case is shown in Figure~\ref{fig:zag-zig-RebalanceZig}. We borrow the terminology from splay trees that if a path branches left, we say it is a zig, and if it branches right it is a zag. If a left branch is followed by a right branch it is a zig-zag. We denote the case in Figure~\ref{fig:zag-zig-RebalanceZig} the zag-zig case. In this case the insertion point moves from being a child of $v$ to being a child of $u$, but retains the same depth, i.e., the insertions and \RebalanceZig increased the depth of the insertion point by one. The probability that $u$ was not rotated up is $p$ and the probability that $v$ is rotated up is $1-p$, i.e., the insertion of $i$ and $n+1-i$ causes the depth of the insertion point to increase by one with probability at least $p(1-p)$. It follows that after the insertion of all $n$ elements, the expected depth of the insertion point is at least $\frac{1}{2}np(1-p)$. 
\end{proof}

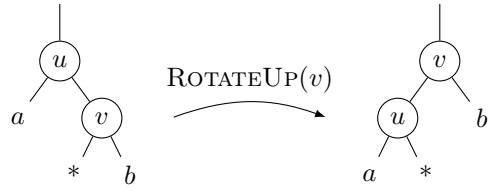
\begin{figure}
  \centering
  \begin{tikzpicture}
    \begin{scope}[xshift=-2.5cm]
      \Tree [ [.{$u$} \node[draw=none] {$a$}; [.{$v$} \node[draw=none] {*}; \node[draw=none] {$b$}; ] ] ]
    \end{scope}
    \draw[-latex] (-1, -1.5) to[bend left=20] node[above, midway] {$\RotateUp{v}$} ++(2, 0);
    \begin{scope}[xshift=2.5cm]
      \Tree [ [.{$v$} [.{$u$} \node[draw=none] {$a$}; \node[missing] {*}; ] \node[draw=none] {$b$}; ] ]
    \end{scope}
  \end{tikzpicture}
  \caption{Bad zag-zig case for {\RebalanceZig}$(v)$, where rotating up $v$ does not decrease the depth of the insertion point *.}
  \label{fig:zag-zig-RebalanceZig}
\end{figure}

If an external leaf has depth~$d$, then the sum of the depths of the $d$ internal nodes on the path to the external leaf is $\sum_{i=0}^{d-1} i=\frac{1}{2}d(d-1)$. The average depth of all nodes in the tree is then at least $\frac{1}{2}d(d-1)/n$, and Theorem~\ref{thm:zig-converging} follows from Lemma~\ref{lem:zig-converging}.

It should be noted that the rotation after an insertion can happen higher in the tree, where there is also a zag-zig case or symmetric zig-zag case, where \RebalanceZig also fails to decrease the depth of the insertion point after the insertion. This explains the gap between the experimental constant observed in Figure~\ref{fig:experiments-comparison} and the theoretical analysis.

\subsection{Pairs Sequences}
\label{sec:zig-pairs}

The pairs sequence consists of $2,1,4,3,6,5,\ldots,n,n-1$. It is essentially an increasing sequence, with pairs $2i-1$ and $2i$ swapped. Pair sequences are not finger sequences. Interestingly, experiments show that \RebalanceZig is challenged by this sequence. In our experimental evaluation, Figure~\ref{fig:experiments-comparison}(a), it appears that $p=\frac{1}{2}$ is a local minima for the average node depth when rebalancing pairs sequences using \RebalanceZig, with increased average node depth for both $p$ smaller than and larger than $\frac{1}{2}$. 
\ShortVersion{In~\cite{fun24arxiv} we prove the following lemma.}%
\LongVersion{In Figure~\ref{fig:zig-pairs} this is quantified as linear for two values of $p\neq \frac{1}{2}$ ($p=\frac{1}{4}$ and $p=\frac{3}{4}$) and about $\sqrt{n}$ for $p=\frac{1}{2}$. Note the very distinct behavior of $p=\frac{1}{4}$ and $p=\frac{3}{4}$: For $p=\frac{1}{4}$ the right path is linear and the left path is constant, whereas for $p=\frac{3}{4}$ it is the opposite.}

\begin{lemma}
  \label{lem:zig-pairs}
  Applying \RebalanceZig to the pairs sequence with $n$ elements, for $n$ even and constant $p\neq \frac{1}{2}$, $0\leq p \leq 1$, the resulting tree has expected average node depth~$\Theta(n)$.
\end{lemma}

\LongVersion{%
\begin{proof}
  For $p=1$, no rebalancing is performed, and the tree is a right path of all even numbers in increasing order, with the odd numbers as offsprings  (see Figure~\ref{fig:unbalanced-insertions}). The average depth is $\Big(\sum_{i=1}^{n/2-1} 2i+n/2\Big)/n=n/4$. For $p=0$, the rebalancing is always performed at the inserted node. It turns that the resulting tree has root 1, with no left child, and $n$ as it right child, that again has no right child, but has the remaining numbers in decreasing order on one leftwards path (see Figure~\ref{fig:pairs-greedy-rotations}). Since this is a single path, the average depth is $(n-1)/2$.

  For the pairs sequence we insert pairs $2i$ and $2i-1$, for $i=1,\ldots,n/2$. We first analyze how the insertion of a pair $2i$ and $2i-1$ influences the right height. Assume the right height before the two insertions is $d$. The change in the right height due to the two insertions is between $-1$ and $+2$, as follows from the following 10 cases. Table~\ref{tab:pairs-probabilities} summarizes the 10 cases. We first insert $2i$, as the new maximum in a new rightmost node. This increasing the right height by one, before applying $\RebalanceZig$, that causes the following cases:
  \begin{itemize}
    \item Inserting $2i$ does not cause any rotations (i.e., the coin flips tail at node $2i$ and existing nodes on the rightmost path, except for the root). This happens with probability $p^{d+1}$. The right height increases by one. Node $2i-1$ is inserted as a left child of $2i$ before applying $\RebalanceZig$. When applying $\RebalanceZig$ after inserting $2i-1$ we have the following subcases:
    \begin{enumerate}
      \item No rotation is performed, because the coin flips tail at $2i-1$, $2i$, and the $d$ previous non-root nodes on the rightmost path. This happens with probability $p^{d+2}$. The rightmost path is not affected.
      \item $2i-1$ is rotated up (above $2i$ on the rightmost path), with probability $1-p$. The right height increases by one.
      \item $2i$ or an ancestor of $2i$ is rotated up, reducing the right height by one (again, after inserting $2i$ increased the height). The happens when the coin flips tail at $2i-1$ and we do not reach the root by $d+1$ additional coin flips showing tail. This happens with probability $p(1-p^{d+1})$.
    \end{enumerate}
    \item Inserting $2i$ rotates $2i$ up with probability $1-p$. The right height remains the same. Node $2i-2$ becomes the left child of $2i$, and $2i-1$ is inserted as a right child of $2i-2$ before applying $\RebalanceZig$. Applying $\RebalanceZig$ causes the subcases:
    \begin{enumerate}
      \setcounter{enumi}{3}
      \item No rotation is performed, because of $d+2$ tail outcomes. This happens with probability $p^{d+2}$, and the right height is unchanged. 
      \item $2i-1$ is rotated up with probability $1-p$, and $2i-1$ becomes a left child of $2i$, and $2i-2$ a left child of $2i-1$. The right height is unchanged.
      \item $2i-2$ is rotated up with probability $p(1-p)$, since we get a tail at $2i-1$ and head at $2i-2$. This causes $2i-2$ to be rotated back on the rightmost path above $2i$, and $2i-1$ becomes the left child of $2i$. The right height increases by one, leaving the tree in the same state as case 1) where no rotations were performed.
      \item A node on the rightmost path is rotated up, and decreases the right height by one. This happens when we have tail at $2i-1$ and $2i-2$, and not tail at the $d$ non-root nodes on the rightmost path. This happens with probability $p^2(1-p^d)$.
    \end{enumerate}
    \item Inserting $2i$ causes a rotation at an ancestor of $2i$. This happens when the coin flips tail at $2i$, and we do not flip additional $d$ tails in a row (otherwise we reach the root that cannot be rotated up). This happens with probability $p(1-p^d)$. The right height is unchanged. Node $2i-1$ becomes a left child of $2i$ before applying $\RebalanceZig$.  Applying $\RebalanceZig$ causes the subcases:
    \begin{enumerate}
      \setcounter{enumi}{7}
      \item No rotation is performed and the right height is unchanged, because of $d+1$ tail outcomes. This happens with probability~$p^{d+1}$.
      \item $2i-1$ rotated up on the path as the parent of $2i$. This happens with probability~$1-p$, and the right height increases by one.
      \item A node on the rightmost path is rotated up, and decreases the right height by one. This happens when we have tail at $2i-1$ and not $d$ additional tails in a row, which happens with probability $p(1-p^d)$.
    \end{enumerate}
  \end{itemize}

  \begin{table}[th]
    \caption{The change in the right height by inserting a pair $2i$ and $2i-1$ using \RebalanceZig, where the right height before the insertions is~$d$. For the change in right weight we state the change in right by insertion $2i$ plus the change by inserting $2i-1$.}
    \label{tab:pairs-probabilities}
    \centering
    \begin{tabular}{clc}
      Case & Probability & Right height \\
      \hline
        1 & $p^{d+1}  \cdot p^{d+2}$      & $+1+0=+1$ \\
        2 & $p^{d+1}  \cdot (1-p)$        & $+1+1=+2$ \\
        3 & $p^{d+1}  \cdot p(1-p^{d+1})$ & $+1-1=+0$ \\
        4 & $(1-p)    \cdot p^{d+2}$      & $+0+0=+0$ \\
        5 & $(1-p)    \cdot (1-p)$        & $+0+0=+0$ \\
        6 & $(1-p)    \cdot p(1-p)$       & $+0+1=+1$ \\
        7 & $(1-p)    \cdot p^2(1-p^d)$   & $+0-1=-1$ \\
        8 & $p(1-p^d) \cdot p^{d+1}$      & $+0+0=+0$ \\
        9 & $p(1-p^d) \cdot (1-p)$        & $+0+1=+1$ \\
       10 & $p(1-p^d) \cdot p(1-p^d)$     & $+0-1=-1$ \\
       \hline
     \end{tabular}
  \end{table}

  Summarizing the expected change in the right height by inserting the pair $2i$ and $2i-1$ in a tree with right height~$d$ using $\RebalanceZig$ is:
  \begin{equation}
    \label{eq:zig-pairs-raw-height}
    \begin{array}l
      -1 \cdot \left((1-p) \cdot p^2\cdot(1-p^d) + p\cdot(1-p^d) \cdot p\cdot(1-p^d) \right) \\
      +0 \cdot \left(p^{d+1} \cdot p\cdot(1-p^{d+1}) + (1-p) \cdot p^{d+2} + (1-p) \cdot (1-p) + p\cdot (1-p^d) \cdot p^{d+1}\right) \\
      +1 \cdot \left(p^{d+1}\cdot p^{d+2} + (1-p)\cdot p\cdot(1-p) + p\cdot(1-p^d)\cdot(1-p)\right) \\
      +2 \cdot \left(p^{d+1}\cdot(1-p)\right) \;.
    \end{array}
  \end{equation}
  This can be reduced to
  \begin{equation}
    \label{eq:zig-pairs-right-height}
    2p^3 - 5p^2 + 2p + p^{d+1}\cdot(1 + 2p - p^2 + (p - 1) p^{d + 1})\;,
  \end{equation}
  where $1\leq 1 + 2p - p^2 + (p - 1) p^{1 + d}\leq 2$ for all $0\leq p\leq 1$ and $d\geq 0$.

  For the case $0<p<\frac{1}{2}$ we prove that the rightmost path has expected length $\Omega(n)$, i.e., the expected average node depth is $\Omega(n)$. By (\ref{eq:zig-pairs-right-height}) the expected increase in the right height by inserting the pair $2i$ and $2i-1$ is at least $2p^3-5p^2+2p=p(p-2)(2p-1)>0$ for $0<p<\frac{1}{2}$. Since we insert $n/2$ pairs, at the end the expected right height is at least $n(p^3-\frac{5}{2}p^2+p)=\Theta(n)$, for constant~$0<p<\frac{1}{2}$.

  For the case $\frac{1}{2}<p<1$, we argue that the leftmost path has expected length $\Omega(n)$, by first arguing that the right height is expected $O(1)$ for each insertion. For $\frac{1}{2}<p<1$, we have $2p^3 - 5p^2 + 2p < 0$ in (\ref{eq:zig-pairs-right-height}). By setting $d_0$ sufficiently large ($d_0 \geq (\lg (-2p^3+5p^2-2p) - 2)/ \lg p - 1$), we have  $2p^{d_0+1}\leq -\frac{1}{2}(2p^3-5p^2+2p)$ and (\ref{eq:zig-pairs-right-height}) is negative for all $d\geq d_0$.

  Letting $p^-$, $p_0$, $p^+$ and $p^{+\!\!+}$ be the probabilities for $-1$, 0, $+1$ and $+2$ in (\ref{eq:zig-pairs-raw-height}) for $d=d_0$ (the probabilities to perform a rotation on the rightmost path increases with the right height), respectively, we have $-p^- + p^+ + 2p^{+\!\!+}<0$. We let $\pi_n^i$ denote the probability that right height is at least $i$ after $n$ pairs have been inserted, that changed the length of the rightmost path. We assume without loss of generality that $p_0=0$ and scale $p^-$, $p^+$ and $p^{+\!\!+}$ by a factor $\frac{1}{1-p_0}$.   Let $\epsilon=p^- - p^+ - 2p^{+\!\!+}>0$. By induction in $n$ we prove that $\pi_n^i \leq c^{i-d_0}$ for $c=1-\epsilon/2$. Note $0< c < 1$. For $i \leq d_0$ this statement is vacuously true, since $c^{i-d_0}\geq 1$, and for $n=0$ we have done no insertions, so $\pi_n^i=0$ for~$i > d_0$. 
    
  For $i>d_0$ and  $n>0$ we have 
  \[
    \pi_n^i = p^{+\!\!+} \cdot \pi_{n-1}^{i-2} + p^+ \cdot \pi_{n-1}^{i-1} + p^- \cdot \pi_{n-1}^{i+1} \; .
  \]
  To show $\pi_n^ i \leq c^{i-d_0}$, it by the induction hypothesis is sufficient to show
  \[
    p^{+\!\!+} \cdot c^{i-2-d_0} + p^+ \cdot c^{i-1-d_0} + p^- \cdot c^{i+1-d_0} \leq c^{i-d_0}\;,
  \]
  or equivalently
  \[
    p^{+\!\!+} + p^+ \cdot c + p^- \cdot c^{3} \leq c^{2}\;.
  \]
  From  $p^-+p^++p^{+\!\!+}=1$ and $p^+ + 2p^{+\!\!+} = p^- - \epsilon$, we have $p^{+\!\!+} = 2p^- - 1 - \epsilon$ and $p^+ = 2 + \epsilon -3p^-$. Inserting this into above, we have
  \[
     (2p^- - 1 - \epsilon) + (2 + \epsilon -3p^-) \cdot c + p^- \cdot c^{3} \leq c^{2}\;,
  \]
  or equivalently
  \[
     p^-\cdot(2-3c+c^3) + (-1 - \epsilon) + (2 + \epsilon) \cdot c \leq c^{2}\;.
  \]
  Since $2-3c+c^3\geq 0$ for $0\leq c\leq 1$, we have that the left side of the above inequality is maximized when $p^-$ is maximized, which is for $p^-=\frac{2+\epsilon}{3}$ (follows from $p^-+p^++p^{+\!\!+}=1$ and $p^+ + 2p^{+\!\!+} = p^- - \epsilon$ by setting $p^+=0$). The equation above holds for all $0\leq \epsilon\leq 1$, $p^- \leq \frac{2+\epsilon}{3}$, and $c=1-\epsilon/2$. It follows that $\pi_n^i \leq (1-\epsilon/2)^{i-d_0}$ and the expected right height is $O(d_0+1)=O(1)$ for all insertions.

  If the expected right height is $O(1)$ for all insertions, then with constant probability a constant fraction $I$ of the insertions has right height $O(1)$ when inserted. We will argue that a constant fraction of $I$ is expected to be on the leftmost path, i.e., the leftmost path has expected length $\Omega(n)$. For a node in $I$ we will consider its lifetime in the tree.

  For a node~$v$ we define its \emph{right depth} to be the number of times the path from the root to $v$ branches right. We define the \emph{left depth} of a node slightly different, it is number of times the path from the root to $v$ branches left, excluding branches on the leftmost path from the root. We define the \emph{branching depth} of a node $v$ to be the depth of the deepest node on the rightmost path to $v$ (possibly $v$, if $v$ is on the rightmost path). When a node from $I$ is inserted it has depth $O(1)$, and its right depth, left depth and branching depth is also $O(1)$. Inserting a pair $2i$ and $2i-1$ performs at most two rotations. Most rotations are left rotations. The only right rotations are in cases 2) and 9), where $2i-1$ is a left child of $2i$ and is rotated onto the rightmost path, and case 6) where the effect of doing the two rotations is the same as doing no rotations at all. It follows that no right rotation can increase the right depth of any node residing in the tree before inserting the pair. Left rotations never increases the right depth. Note $\RotateUp{u}$ in Figure~\ref{fig:rotation} the right depth of $u$ and nodes in its right subtree~$c$ have their right depth decreased by one, whereas the right depth of the nodes $v$ and in subtress $a$ and $b$ remain unchanged.  Node~$v$ and nodes in subtree~$a$ have their left depth increased by one. If $u$ is a node on the rightmost path, the branching level of all nodes in the subtree of $u$ is decreased by one, including $u$ itself. It follows that the right depth stays $O(1)$ after insertion, and the branching level of a node can decrease by a left rotation on the rightmost path, but never increase again. 
  
  We next argue that the left depth of a node in $I$ is expected $O(1)$, i.e., with constant probability a constant fraction $J\subseteq I$ also have left depth $O(1)$. If the bounds on the left and right depths are $\ell$ and $r$, respectively, it follows that each node on the leftmost path together with its right subtree can have at most $2^{\ell+r}$ nodes from $J$. It follows that the leftmost path must have length at least $|J|/2^{\ell+r}=\Omega(n)$ nodes, and overall that the average node depth is~$\Omega(n)$ (the hidden constant depends on the probability parameter $p$; we are oblivious to this here).

  The left depth increases because of left rotations on the rightmost path, or because $2i-1$ is the right child of $2i-2$ and $2i-1$ is rotated up, case 5). Assume there is a rotation on the rightmost path, that causes a node $u$ on the rightmost path to be rotated up, causing the left depth to increase by one for its parent~$v$ (before the rotation) and the nodes in the left subtree of~$v$. This is because a coin turned out head at $u$ with probability $1-p$. In the cases where a coin is flipped at the right child of $v$, with probability $p(1-p)$ we instead would have rotated $v$ up (provided $v$ is not the root), i.e., we will  flip a coin  expected $\frac{1}{p(1-p)}$ times at the right child of $v$ before $v$ is rotated up, and its branching level decreases by one. Since the branching level is bounded by the initial depth, that is $O(1)$ for nodes in $I$, we get that the expected number of times the left depth increases for a node in $I$ is $O\Big(\frac{1}{p(1-p)}\Big)=O(1)$. Similarly in case~5), where we flip a coin $2p-1$ and rotate it up with probability~$1-p$, there is probability~$pp(1-p)$ that we instead would have rotated $2i$ up, i.e., a subcase of case 7), where the branching level is decreased. It follows that 5) can only increase the left depth of a node expected $O(1)$ times before the branching level decreases by one. This completes the case~$\frac{1}{2}<p<1$.
\end{proof}
}

The last inserted element has expected depth $\Theta(n)$ for $0<p<\frac{1}{2}$ and $O(1)$ for $\frac{1}{2}<p<1$, so Lemma~\ref{lem:zig-pairs} does not give any bounds on the expected depth of specific elements. Lemma~\ref{lem:zig-pairs} addresses pairs sequences for $p\neq \frac{1}{2}$, where the expected average node depth is linear. For $p=\frac{1}{2}$ we give the following conjecture, stating that the complexity is significantly different.
\ShortVersion{See~\cite{fun24arxiv} for an experimental and theoretical motivation of the conjecture.}

\begin{conjecture}
  \label{conj:zig-pairs}
  Applying \RebalanceZig to the pairs sequence with $n$ elements, for $n$ even and $p=\frac{1}{2}$, the resulting tree has expected average node depth~$O\left(\sqrt{n}\right)$.
\end{conjecture}

\LongVersion{%
The first basis for this conjecture are our experiments. In our simulation, Figure~\ref{fig:zig-pairs}, the average node depth for $p=\frac{1}{2}$ tends to be about $\sqrt{n}$. Secondly, intuitively the change to the length of the rightmost path resembles the simple process studied in Theorem~\ref{thm:process} (but that does not include a $p^{+\!\!+}$ case). Here the expected value is $\Theta\left(\sqrt{n}\right)$, under the assumption~$p^-=p^+>0$. For $p=\frac{1}{2}$, the expected increase in right height by inserting a pair is given by (\ref{eq:zig-pairs-right-height}), that has value between 0 and $3p^{d+1}$. For $d\geq \log_{1/p} n$, the value $3p^{d+1}\leq 3/n$, so the additional increase to the expected right height is only $O(1)$ due to this term if the right height is~$\Omega(\log_{1/p} n)$. How does the right height relate to the average node of all nodes? For $\frac{1}{2}<p<1$, we saw that the tree tended to be a leftmost path because of $\Theta(n)$ rotations at the root. But for $p=\frac{1}{2}$ we do not expect to reach the root that often (if Theorem~\ref{thm:process} applied, it would be $\Theta(\sqrt{n})$). Essentially, we would expect that the right height when a node is created is a good estimate of the nodes final depth, since rotations on the rightmost path have about equal probability to reduce and increase the depth of a node.
}

\section{Algorithm {\RebalanceZigZag}}
\label{sec:zigzag}

To address the shortcomings of algorithm \RebalanceZig in the case where $v$ is in a zig-zag or zag-zig state, we borrow terminology from splay trees~\cite{SleatorTarjan85}, and apply the zig-zag transformation to the tree (see Figure~\ref{fig:zigzag}(right) and~\cite[Figure~3]{SleatorTarjan85}) by rotating up $v$ twice. In the zig-zig case, instead of rotating up $v$, we rotate up the parent of $v$ (see Figure~\ref{fig:zigzag}(left)). These two transformations ensure that everything in the subtree of $v$ is moved one level up in the tree when applying the transformantion at node~$v$. The pseudo-code for algorithm \RebalanceZigZag is shown in Algorithm~\ref{alg:RebalanceZigZag}.

\begin{algorithm}
  \caption{\textsc{RebalanceZigZag}$(v)$}
  \label{alg:RebalanceZigZag}
  \begin{algorithmic}
    \WHILE {$v.p \neq \nil$ \AND coin flip is tail}
      \STATE $v \gets v.p$
    \ENDWHILE
    \IF {$v.p \neq \nil$ \AND $v.p.p \!\neq\! \nil$}
      \IF {$(v \!=\! v.p.l$ \AND $v.p \!=\! v.p.p.l)$ \OR $(v \!=\! v.p.r$ \AND $v.p \!=\! v.p.p.r)$}    
        \STATE{$\RotateUp{v.p}$}  \hfill \textit{$\triangleright$ zig-zig or zag-zag case}      
      \ELSE
        \STATE{$\RotateUp{v}$}  \hfill \textit{$\triangleright$ zig-zag or zag-zig case} \\
        \STATE{$\RotateUp{v}$}
      \ENDIF
    \ENDIF
  \end{algorithmic}
\end{algorithm}

\subsection{Increasing Sequences}
\label{sec:zigzag-increasing}

\RebalanceZigZag handles increasing and decreasing sequences identical to \RebalanceZig, except that rotations happen one level higher, i.e., the trees are identical if ignoring the rightmost inserted node. Equivalently, this corresponds to inserting the next element without rebalancing, and first performing the rebalancing just before inserting the next element. From Theorem~\ref{thm:zig-increasing} we have the following corollary. 

\begin{corollary}
  \label{cor:zigzag-increasing}
   For $0<p<1$, after inserting $n$ elements in increasing or decreasing order using \RebalanceZigZag, each node has expected depth $O\big(1/p\cdot\log_{1/p} n\big)$.
\end{corollary}

\subsection{Finger Sequences}
\label{sec:zigzag-finger}

We will prove that using \RebalanceZigZag to rebalance finger sequences (like increasing, decreasing and converging sequences), ensures that the resulting tree is expected to be balanced, for $p$ sufficiently large. Recall that a finger sequence is defined such that the next element is always the immediate predecessor or successor of the most recently inserted element among all elements inserted so far. In an unbalanced search tree, this means that the next node will be a (left or right) child of the most recently inserted node, i.e., the resulting tree is always a path. We denote the external leaf where to create the next node the \emph{insertion point}. The crucial property of the restructuring done by \RebalanceZigZag in Figure~\ref{fig:zigzag} is that if the insertion point is at an external leaf in the subtree rooted at $v$ before the rotation, then the depth of this external leaf is reduced by exactly one in both cases. 

\begin{lemma}
  \label{lem:zigzag-finger}
  After inserting a finger sequence with $n$ elements using \RebalanceZigZag, each node has expected depth $O(\lg n)$, for $\frac{1}{2} \left(\sqrt{5} - 1\right)<p<1$.
\end{lemma}

\begin{proof}
  The proof follows the same idea as in Section~\ref{sec:zig-increasing} for the analysis of \RebalanceZig on increasing sequences. Instead of right height we consider \emph{insertion depth}, i.e., the depth $d$ of the parent node of the insertion point.
\ShortVersion{See~\cite{fun24arxiv} for proof details.}%
\LongVersion{%
  Similar to Lemma~\ref{lem:zig-right-height-increase}, \RebalanceZigZag ensures that the insertion depth stays $d$ or increases to $d+1$ with probability $1-p^d$ and $p^d$, respectively. Note the change from $p^{d+1}$ to $p^d$, since we can only apply the transformations in Figure~\ref{fig:zigzag} at nodes with at least two ancestors, whereas Lemma~\ref{lem:zig-right-height-increase} only required at least one ancestor. Similar to Lemma~\ref{lem:zig-right-height}, we can prove that the insertion depth   is at most $\ceil{(c+1)\cdot\log_{1/p} n}$ with probability $1-1/n^c$, for any constant $c > 0$. The only difference compared to the proof of Lemma~\ref{lem:zig-right-height} is that we have the constraint $np^d \leq n^{-c}$ (instead of $np^{d+1} \leq n^{-c}$), but the same bound is a solution to this constraint.

  Finally, we need to rephrase Lemma~\ref{lem:zig-increasing} in terms of \emph{insertion path}, i.e., the path from the insertion point to the root, instead of the rightmost path. Similarly as in the proof of  Lemma~\ref{lem:zig-increasing}, a node $v$ goes through five phases: 1) not inserted yet; 2) created as a leaf; 3) on insertion path (can monotonically be moved closer to root by rotations at ancestors); 4) rotated out of insertion path; and 5) pushed up or down by rotations on the insertion path. Given the high probability bound on the insertion depth, after phases 1)--4) $v$ has depth $O\big(1/p\cdot\log_{1/p} n\big)$, with high probability. If a node is pushed down in the tree during phase~5), it is because it resides in the offspring~$d$ to the insertion path in both cases in Figure~\ref{fig:zigzag}, and the coin turned out head with probability~$1-p$ at the grandchild~$v$ of the branching node $w$ to $d$ on the insertion path. Similarly, the depth of $v$ decreases exactly when the rebalancing is performed at $w$ or an ancestor of~$w$. This happens with probability~$p^2$ if a coin is flipped at both $v$ and $u$ showing tail (with probability $p(1-p)$ the coin shows tail at $v$ and head at $u$, where the rebalancing is performed at $u$, but then the depth of the subtree $d$ is unchanged). It follows that the depth of node not on the insertion path can increase with probability~$p^+=1-p$ and decrease with probability $p^-=p^2$ for those insertions flipping a coin at the grandchild~$v$ on the insertion path of the branching node $w$ to $d$. There are two cases that we safely can ignore: If the branching node is the parent of the insertion point, then the rebalancing can decrease the depth of the node, but never increase it. The other case is if the rebalancing reaches the root without rotations. In the above we assumed this would decrease the depth of the node with probability $p^-$, but the depth of the node remains unchanged. But since the depth of the insertion path increases in this case, we have an upper bound of $O\big(1/p\cdot\log_{1/p} n\big)$ on how often this can happen, and can contribute to an increased depth of the node. 
  
  We will use Lemma~\ref{lem:exp-decrease} to bound the additional increase in the depth of a node. Since we assume $p>\frac{1}{2} \left(\sqrt{5} - 1\right)$, we have $p^+=1-p<p^2=p^-$ and the lemma applies. Actually, the number of times we flip a coin at a grandchild of the branching node is a stochastic variable depending on the tree structure (Lemma~\ref{lem:exp-decrease} assumes we consider exactly $n$ flips), so the lemma does not completely apply. But from the proof of the lemma it follows that, since $p^+<p^-$, the probability that the depth increases by $\Delta$ is exponentially decreasing in $\Delta$, i.e., with high probability the increase in depth is at most $O(\lg n)$ for up to $n$ insertions. 
}
\end{proof}

In our experiments, see Figure\ref{fig:experiments-comparison}(c), it shows up that \RebalanceZigZag performs better for $p>\frac{1}{2}$ than $p<\frac{1}{2}$ on the converging sequence, that is an example of a finger sequence. We leave open the question what the dependency on $p$ is for \RebalanceZigZag for $0<p\leq \frac{1}{2} \left(\sqrt{5} - 1\right)$.

\begin{figure}
  \centering
  \begin{tikzpicture}
    \begin{scope}[xshift=-1.8cm]
      \Tree [ .{$w$} [.{$u$} [.{$v$} \node[draw=none] {$a$}; \node[draw=none] {$b$}; ] \node[draw=none] {$c$}; ] \node[draw=none] {$d$}; ]
    \end{scope}
    \draw[-latex] (-1, 0) to[bend left=20] node[above, midway,align=center] {\small $\RotateUp{v.p}$} ++(2, 0);
    \begin{scope}[xshift=1.6cm, yshift=-3ex]
      \Tree [ .{$u$} [.{$v$} \node[draw=none] {$a$}; \node[draw=none] {$b$}; ] [.{$w$} \node[draw=none] {$c$}; \node[draw=none] {$d$}; ] ]
    \end{scope}
  \end{tikzpicture}
  \hfill
  \begin{tikzpicture}
    \begin{scope}[xshift=-1.8cm]
      \Tree [ .{$w$} [.{$u$} \node[draw=none] {$a$}; [.{$v$} \node[draw=none] {$b$}; \node[draw=none] {$c$}; ] ] \node[draw=none] {$d$}; ]
    \end{scope}
    \draw[-latex] (-1, 0) to[bend left=20] node[above, midway,align=center] {\small $\RotateUp{v}$ \\ \small $\RotateUp{v}$} ++(2, 0);
    \begin{scope}[xshift=1.5cm, yshift=-3ex]
      \Tree [ .{$v$} [.{$u$} \node[draw=none] {$a$}; \node[draw=none] {$b$}; ] [.{$w$} \node[draw=none] {$c$}; \node[draw=none] {$d$}; ] ]
    \end{scope}
  \end{tikzpicture}
  \caption{The rebalancing performed by algorithm {\RebalanceZigZag}$(v)$ in  (left) zig-zig case and (right) zig-zag case.}
  \label{fig:zigzag}
\end{figure}
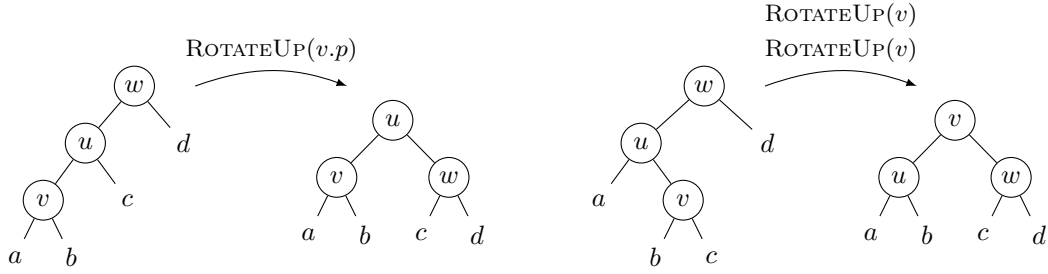

\subsection{Pairs Sequences}
\label{sec:zigzag-pairs}

While \RebalanceZigZag achieves better average node depth on converging sequences compared to \RebalanceZig, it fails on pairs sequences, where the expected average node depth is linear for all values $0 \leq p \leq 1$ (for $p=\frac{1}{2}$ this is worse than \RebalanceZig, if our conjecture turns out to be true).
\ShortVersion{In~\cite{fun24arxiv} we prove the following lemma.}

\begin{lemma}
  \label{lem:zigzag-pairs}
  Applying \RebalanceZigZag to the pairs sequence with $n$ elements, for $n$ even and $0 \leq p \leq 1$, the resulting tree has expected average node depth~$\Theta(n)$.
\end{lemma}

\LongVersion{%
\begin{proof}
  For $p=1$, no rotations are performed and the rightmost path consists of all even numbers (Figure~\ref{fig:unbalanced-insertions}), and the average node depth is $\Theta(n)$. If $p=0$, the rebalancing is always performed at the inserted node, and the resulting tree will be a leftmost path with $n-1$ nodes $1, 2, \ldots, n-1$, and $n$ as the right child of the root $n-1$ (see Figure~\ref{fig:pairs-greedy-rotations}). Again, the average node depth is~$\Theta(n)$.

  For $0<p<1$, when inserting the pair $2i$ and $2i-1$, $2i$ is first inserted as a rightmost child, and an ancestor of $2i$ is potentially rotated up (as a zag-zag case) with probability~$1-p^d$, where $d$ is the right height before inserting the pair. The right height after inserting $2i$ with rebalancing is $d$ with probability~$1-p^d$ and height $d+1$ with probability~$p^d$. After inserting~$2i-1$ as the left child of $2i$, the rebalancing can happen at node $2i-1$ (zag-zig case) with probability~$1-p$, where $2i-1$ replaces $2i-2$ on the rightmost path. The right height is unchanged. If no rebalancing is done, the height remains unchanged, with probability~$p^{d'}$, where $d'$ is the right before inserting $2i-1$ ($d'=d$ or $d'=d+1$). Finally, the right height can be reduced to $d'-1$, if a rebalancing is performed at $2i$ or an ancestor, that happens with probability~$p(1-p^{d'-1})$. We let $p^+$ and $p^-$ denote the probability, that the right height increases and decreases by one due to the insertion of the pair, respectively. Since the height only increases when no rotation is performed when inserting $2i$, and when inserting $2i-1$ no rotation is performed or the rebalancing happens at $2i-1$. We have
  \[
    p^+ = p^d\cdot\big((1-p) + p^{d+1})\big) \; .
  \]
  For the right height to decrease, rotations on the rightmost path must happen both during the insertion of $2i$ and $2i-1$, that happens with probability
  \[
    p^- = (1-p^d) \cdot p(1-p^{d-1}) \; .
  \]
  For $p^{d-1}\leq \frac{1}{4}$, i.e., $d \geq 1 + \frac{2}{\lg (1/p)}$, we have $p^+\leq\frac{1}{4}p$ and $p^-\geq\frac{9}{16}p$, and the expected change in right height is $p^+ - p^-\leq \left(\frac{1}{4} - \frac{9}{16}\right)p = -\frac{5}{16}p$. Similarly to Lemma~\ref{lem:zig-pairs}, applying Theorem~\ref{thm:process} it follows that the expected right height is $O(1)$. Since the rotations performed during rebalancing only increases the depth of nodes by left rotations, and this happens at most expected $O(1)$ times for each branching level of the node, before the branching level is decreased, we have that with constant probability a constant fraction of the inserted elements have constant left and right depth (using the definition of left depth from Lemma~\ref{lem:zig-pairs}, where left branches on the leftmost path from the root do not count towards the left depth of a node). It follows that with constant probability the leftmost path has length $\Omega(n)$, i.e., the expected average node depth is~$\Theta(n)$.
\end{proof}

\begin{figure}
  \centering
    \tree{[.1 \edge[missing]; \node[missing] {}; [.6 [.5 [.4 [.3 [.2 ] \edge[missing]; \node[missing] {}; ] \edge[missing]; \node[missing] {}; ] \edge[missing]; \node[missing] {}; ] \edge[missing]; \node[missing] {}; ] ]}
    \hspace{5em}
    \tree{[.5 [.4 [.3 [.2 [.1 ] \edge[missing]; \node[missing] {}; ] \edge[missing]; \node[missing] {}; ] \edge[missing]; \node[missing] {}; ] [.6 ] ]}  
  \caption{Applying \RebalanceZig (left) and \RebalanceZigZag (right) on the pairs sequence $2,1,4,3,6,5$ with $p=0$.}
  \label{fig:pairs-greedy-rotations}
\end{figure}
}

\section{Random Permutations}
\label{sec:random-permutation}

We do not prove anything for random permutations, for neither \RebalanceZig nor \RebalanceZigZag. If $p=1$ no rebalancing is performed, and it is known that the expected depth of a node is $O(\lg n)$~\cite{BoothCollin60,Hibbard62,Windley60}, whereas for $p=0$, a rotation is always performed at the inserted leaf, and the tree will always be a single path. How exactly the average node depth depends on $p$ is an open problem. In the experiments, see
\ShortVersion{\cite{fun24arxiv}}%
\LongVersion{Figure~\ref{fig:experiments-permutations}}%
, it appears that \RebalanceZigZag is ``about'' logarithmic for $p \geq 0.7$. 

\section{Experimental Evaluation}
\label{sec:experiments}

\ShortVersion{In this section we present an experimental evaluation of our algorithms \RebalanceZig and \RebalanceZigZag. See~\cite{fun24arxiv} for more experimental results.}

\LongVersion{%
In this section we present an experimental evaluation of our two algorithms \RebalanceZig and \RebalanceZigZag, and also consider a third algorithm \RebalanceZigZig. Algorithm \RebalanceZigZig is identical to \RebalanceZigZag, except in the zig-zig and zag-zag cases. Algorithm \RebalanceZigZag in these cases only performs one rotation. In algorithm \RebalanceZigZig we instead perform two rotations, by first rotating up the parent of the node, and then the node itself --- like the zig-zig and zag-zag cases in splay trees~\cite[Figure~3]{SleatorTarjan85}. See Figure~\ref{fig:zigzig}.

\begin{figure}
  \centering
  \begin{tikzpicture}
    \begin{scope}[xshift=-3cm]
      \Tree [ .{$w$} [.{$u$} [.{$v$} \node[draw=none] {$a$}; \node[draw=none] {$b$}; ] \node[draw=none] {$c$}; ] \node[draw=none] {$d$}; ]
    \end{scope}
    \draw[-latex] (-1, -1.5) to[bend left=20] node[above, midway,align=center] {\small $\RotateUp{v.p}$ \\ \small $\RotateUp{v}$} ++(2, 0);
    \begin{scope}[xshift=3cm]
      \Tree [ .{$v$} \node[draw=none] {$a$}; [.{$u$} \node[draw=none] {$b$}; [.{$w$} \node[draw=none] {$c$}; \node[draw=none] {$d$}; ] ] ]
    \end{scope}
  \end{tikzpicture}
  \caption{The rebalancing performed by algorithm {\RebalanceZigZig} at $v$ in the zig-zig case.}
  \label{fig:zigzig}
\end{figure}
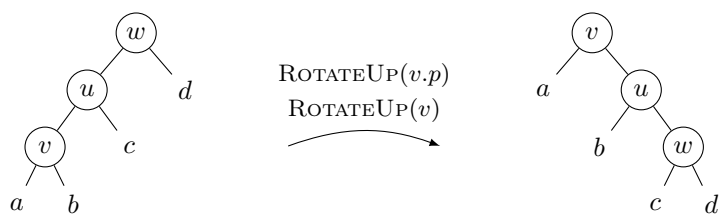
}

We implemented the algorithms in Python~3.12, and ran the algorithms with different choices of $p$ on the types of insertion sequences listed in Table~\ref{tab:sequences} with sequence lengths being powers of two. Each data point in 
\ShortVersion{Figure~\ref{fig:experiments-comparison}}%
\LongVersion{Figures~\ref{fig:experiments-comparison}, \ref{fig:experiments-permutations}, \ref{fig:depth-profiles}}
is the average over 25~runs. The increasing, decreasing and converging sequences are examples of finger insertions. Inserting the pairs, bitonic and runs sequences into a search tree without rebalancing result in identical search trees (see Figure~\ref{fig:unbalanced-insertions}). Note that the first half of the bitonic sequence is an increasing sequence, whereas the second part evenly distributes the remaining elements into the created leaves right-to-left. The runs sequences is identical to the bitonic sequences, except that the second part is performed left-to-right.

Figure~\ref{fig:experiments-comparison} shows our main experimental findings. It shows the resulting average node depths of running the algorithms on the different types of insertion sequences from Table~\ref{tab:sequences} with insertion sequences of length between two and 1024 and various values of $p$ in the range zero to one. Note that for a path with $n$ nodes, the root has depth~0 and the bottommost node depth $n-1$, i.e., the average depth is $\frac{1}{n}\sum_{d=0}^{n-1}=\frac{1}{n}\cdot \frac{n(n-1)}{2}=\frac{n-1}{2}$. For $n=1024$, an average node depth of 511.5 implies that the tree is a path. In Figure~\ref{fig:experiments-comparison}(a) this explains why all curves share the top-left point, since \RebalanceZig always generates a path when $p=0$, independent of the insertion sequence, as discussed in Section~\ref{sec:zig}. In Figure~\ref{fig:experiments-comparison}(left) the rightmost data point ($p=1$) for random permutations corresponds to the average node depth in unbalanced binary search trees. Note that the pairs, bitonic and runs insertion sequences end up with different average node depth characteristics for each of the three algorithms (dashed curves in Figure~\ref{fig:experiments-comparison}), even that they would generate the same trees without rebalancing.

Figure~\ref{fig:experiments-comparison}(a, b) clearly shows that \RebalanceZig has problems with the converging sequences (consistent with Theorem~\ref{thm:zig-converging}); 
Figure~\ref{fig:experiments-comparison}(c, d) that \RebalanceZigZag has problems with the pairs sequences (consistent with Theorem~\ref{thm:zigzag-pairs})%
\ShortVersion{.}%
\LongVersion{; and Figure~\ref{fig:experiments-comparison}(e, f) that \RebalanceZigZig has problems with increasing (and many other) sequences. An intuitive argument why \RebalanceZigZig does not achieve good performance on increasing subsequences is that the insertion point is always the rightmost leaf and all nodes on the path are zag-zag cases. In zag-zag cases, algorithm \RebalanceZigZig performs two rotations moving the insertion point one level towards the root by the insertion, i.e., the insertion point is moved towards the root and all rotations will be left rotations close to the root and likely pushing many elements one level further down.}

\LongVersion{%
Figure~\ref{fig:experiments-permutations} focuses on inserting random permutation sequences, where it is known that no rebalancing gives trees with expected logarithmic node depths~\cite{BoothCollin60,Hibbard62,Windley60}. Here we consider $n$ up to $2^{20}$ and $\frac{1}{2} \leq p \leq 1$. The average node depth decreases with increasing $p$, and for about $p \geq 0.6$ \RebalanceZigZag achieves the best performance of the three algorithms.

Figure~\ref{fig:depth-profiles} shows node depth profiles of the three algorithms. Each curve is the sum of generating 100 trees and computing the total number of nodes with each depth. The plot again shows how \RebalanceZig and \RebalanceZigZag have problems handling converging and pairs sequences, respectively.
}

\section{Conclusion and Open Problems}

This paper leaves more open problems than it solves. None of the considered randomized rebalancing algorithms meets all conditions (\ref{req:no-balance})--(\ref{req:depth}) introduced in Section~\ref{sec:introduction}. Inspired by a question raised by Seidel~\cite{Seidel09}, we considered bottom-up randomized rebalancing schemes for binary search trees without storing any balance information. We studied randomized rebalancing strategies, inspired by the rebalancing primitives from splay trees~\cite{SleatorTarjan85}. They meet conditions (\ref{req:no-balance})--(\ref{req:random-bits}), but fail to achieve logarithmic depth on all insertion sequences. In the experiments \RebalanceZigZag appears often to have the best performance, although it provably does not achieve expected logarithmic average depth for all insertion sequences. It remains an open problem if a randomized bottom-up rebalancing scheme exists that can guarantee expected logarithmic average node depths for all insertion sequences and satisfies requirements (\ref{req:no-balance})--(\ref{req:random-bits}), or what the best depth guarantee can be given requirements (\ref{req:no-balance})--(\ref{req:random-bits}), or how much these requirements need to be relaxed to enable expected logarithmic average node depths. 

\LongVersion{We did not consider deletions at all in this paper (see~\cite{CulbersonMunro89,CulbersonMunro90,Panny10} for challenges on performing deletions in random binary search trees).}

\bibliographystyle{plainurl}
\bibliography{references}

\LongVersion{%
\appendix

\section{A Random Process}
\label{sec:process}

In this section we study the mathematics of the stochastic process modeling that the depth of the insertion point can both increase and decrease during insertions, depending on the rotations performed by the rebalancing.
 
Assume we have independent stochastic variables $X_1, X_2, \ldots, X_n$, defined by
\[
  X_i = \left\{
    \begin{array}{rl}
      -1 & \mbox{with probability $p^-$} \\
      0  & \mbox{with probability $p_0$} \\
      +1 & \mbox{with probability $p^+$}\;, \\
    \end{array}
  \right.
\]
where $p^-+p_0+p^+ = 1$ and $p^+, p_0, p^- \geq 0$. We define stochastic variables $Y_0,\ldots,Y_n$ by 
\[
  Y_i = \left\{
    \begin{array}{cl}
      0 & \mbox{for } i = 0 \\
      \max\{0,Y_{i-1}+X_i\} & \mbox{for } i > 0 \;.
    \end{array}
  \right.
\]

Our goal is to estimate the expected value $\E{Y_{n}}$. In our applications $Y_n$ is typically the depth of the insertion point in a binary search tree after $n$ insertions.
Depending on the relative values of $p^-$ and $p^+$, there are three substantially different cases to consider:
$p^- < p^+$, i.e., it is more likely to get $+1$ than~$-1$ (Lemma~\ref{lem:exp-increase});
$p^- = p^+$, i.e., $-1$ and $+1$ are equally likely (Lemma~\ref{lem:exp-stable}); and
$p^- > p^+$, i.e., $-1$ is more likely than~$+1$ (Lemma~\ref{lem:exp-decrease}).

In the following we let $Z_n = \sum_{i=1}^{n} X_i$ and $M_n = \min_{i=0}^{n} Z_i\leq 0$, assuming $Z_0=0$. Note that due to the maximum with zero in the definition of $Y_n$, we have $Y_n\geq Z_n$ and we can lower bound $\E{Y_n}\geq\E{Z_n}=n\cdot(p^+-p^-)$ (that obviously is negative for $p^->p^+$, but by definition $Y_n \geq 0$). 

The following theorem summarizes the results of Lemma~\ref{lem:exp-increase}, Lemma~\ref{lem:exp-stable}, and Lemma~\ref{lem:exp-decrease} (ignoring the dependencies on the constants $p^-$, $p_0$ and $p^+$).

\begin{theorem}
  \label{thm:process}
  The expected value of $Y_n$ is 
  \[
    \E{Y_n} = \left\{
      \begin{array}{cl}
        \Theta(1)        & \mbox{for~~} p^- > p^+ \\
        \Theta\big(\sqrt{n}\big) & \mbox{for~~} p^- = p^+ > 0 \\
        \Theta(n)        & \mbox{for~~} p^- < p^+ \; ,
      \end{array}
    \right.
  \]
  for constants $p^-, p_0, p^+ \geq 0$ and $p^-+p_0+p^+ = 1$.
\end{theorem}

The first case we consider is where $+1$ is more likely than $-1$, i.e.,  the sequence $Y_1,\ldots,Y_n$ tends to be increasing.

\begin{lemma}
  \label{lem:exp-increase}
  For $p^- < p^+$, we have $n \cdot (p^+ - p^-) \leq \E{Y_n} \leq n$.
\end{lemma}

\begin{proof}
  Follows from $Y_n \leq n$ and $\E{Y_n} \geq \E{\sum_{i=1}^{n} X_i} = \sum_{i=1}^{n} \E{X_i} = n\cdot(p^+-p^-)$.
\end{proof}

Next, we consider the case when $+1$ and $-1$ are equally likely, i.e., $\E{X_i}=0$ and $\E{Z_n}=0$. 

\begin{lemma}
  \label{lem:exp-stable}
  For $p^- = p^+ > 0$, we have $\E{Y_n} = \Theta\big(\sqrt{n}\big)$.
\end{lemma}

\begin{proof}
  If $p_0>0$, we first simplify the analysis by removing all $X_i=0$. Let $S_n$ denote the number of $X_i \neq 0$, for $i=1,\ldots,n$. We have $S_n$ is a binomial distribution with parameter $p^-+p^+$ and $\E{S_n} = n\cdot (p^- + p^+)$. By the Chernoff bound (\cite[Theorem 1.1]{DubhashiPanconesi09}) we have that $\Prob{S_n \geq n \cdot (p^- + p^+ - \epsilon)} \geq 1-e^{-2n\epsilon^2}$ for $\epsilon > 0$. It follows that with high probability $S_n=\Theta(n)$, and it is sufficient to show that $\E{Y_n} = \Theta\big(\sqrt{n}\big)$ for $p_0 = 0$ and $p^+=p^-=\frac{1}{2}$ ($X_i$ is a Rademacher distribution~\cite{Montgomery90}).

  The question we need to address, is how often do we expect $Y_{i-1}+X_{i}$ to be negative, for $1\leq i \leq n$. This happens exactly $-\min(Z_0, Z_1, \ldots, Z_n)=-M_n$ times, i.e., we have $Y_n = Z_n - M_n$ and $\E{Y_n} = \E{Z_n - M_n} = \E{Z_n} - \E{M_n} = -\E{M_n}$.

  We first prove $\E{M_n}=-\Omega\big(\sqrt{n}\big)$. The anti concentration bound for the Rademacher distribution states that $\Prob{|Z_n| < \epsilon\sqrt{n}} = O(\epsilon)$, for all $\epsilon > 0$. See the survey by Nguyen and Vu~\cite[Theorem~1.2]{NguyenVu13} or the original paper by Erd\H{o}s~\cite{Erdos45}. We have $\Prob{|Z_n| < \epsilon\sqrt{n}} \leq \frac{1}{2}$, for some $\epsilon > 0$. By the symmetry of $S_n$ we have $\Prob{Z_n \leq - \epsilon \sqrt{n}} \geq \frac{1}{4}$. Since $M_n \leq \min(0, Z_n)$, it follows that $\E{Y_n}=-\E{M_n} \leq -\frac{1}{4}\epsilon\sqrt{n}$.

  Next we prove $\E{M_n} \geq -O\big(\sqrt{n}\big)$. We can drop $Z_0=0$ from the definition of $M_n$, since each $\E{Z_i}=0$, implying the expectation of their minimum cannot be positive. We will make use of the Chernoff bound~\cite{Montgomery90} 
  \begin{equation}
    \label{eq:chernoff}
    \Prob{Z_n < -\epsilon\sqrt{n}} \leq e^{-\epsilon^2/2}, 
  \end{equation}
  for all $\epsilon>0$. From the union bound, we get
  \[
    \Prob{M_n < -\epsilon\sqrt{n}} 
    \leq \sum_{i=1}^n \Prob{Z_i < -\epsilon\sqrt{n}} 
    \leq \sum_{i=1}^n \Prob{Z_i < -\epsilon\sqrt{i}} 
    \leq n \cdot e^{-\epsilon^2/2} \; .
  \]
  For $c \geq 2$ and $\epsilon=c\sqrt{\ln n}$, we get $\Prob{M_n < - c \sqrt{\ln n} \sqrt{n}} \leq n^{1-c^2/2} \leq \frac{1}{n}$. To get a tighter bound on $M_n$ we need to take the dependencies of $Z_i$ into account.

  For $n$ being powers of two and $\epsilon>0$,  we define depth bounds $\mathrm{depth}(n,\epsilon)<0$ and probabilities $\mathrm{prob}(i,\epsilon) \geq 0$ recursively, such that 
  \begin{equation}
    \label{eq:depth-bound}
    \Prob{M_n < \mathrm{depth}(n,\epsilon)} \leq \mathrm{prob}(n,\epsilon)\;.
  \end{equation}
  
  We let $1 < \tau < \sqrt{2}$ be a constant, that controls the slack in the recursion. Our recursive defined bounds are
  \[
    \mathrm{depth}(n,\epsilon) = \left\{
      \begin{array}{ll}
        -1 & \mbox{for } n = 1 \\
        \mathrm{depth}\left(\frac{n}{2}, \epsilon\tau\right) - \epsilon\sqrt{n/2}  & \mbox{otherwise\;,}
      \end{array}
    \right.
  \]
  \[
    \mathrm{prob}(n,\epsilon) = \left\{
      \begin{array}{ll}
        0 & \mbox{for } n = 1 \\
        2 \cdot \mathrm{prob}\left(\frac{n}{2}, \epsilon\tau\right)  + e^{-\epsilon^2/2} & \mbox{otherwise\;.}
      \end{array}
    \right.
  \]
  
  We now prove by induction in $n$ (i.e., the powers of two) that (\ref{eq:depth-bound}) holds with our definition of $\mathrm{depth}(n,\epsilon)$ and $\mathrm{prob}(n,\epsilon)$, for all values of $\epsilon>0$. For the base case $n=1$, we have $M_n=\min(0,X_1)\geq -1$, and $\Prob{M_n <-1}=0$, independently of $\epsilon$. For $n>1$, we view the sequence $Z_1,\ldots,Z_n$ as two parts, the left part $Z_1,\ldots,Z_{n/2}$ and the right part $Z'_{n/2+1},\ldots,Z'_n$,
  where $Z'_i=Z_i-Z_{n/2}=\sum_{j=n/2+1}^i X_j$. For the two parts of length~$n/2$ we can use the induction hypothesis (\ref{eq:depth-bound}), and have $\Prob{\min_{i=1}^{n/2} Z_i < \mathrm{depth}\left(\frac{n}{2},\epsilon\tau\right)} \leq \mathrm{prob}\left(\frac{n}{2},\epsilon\tau\right)$ and $\Prob{\min_{i=n/2+1}^{n} Z'_i < \mathrm{depth}\left(\frac{n}{2},\epsilon\tau\right)} \leq \mathrm{prob}\left(\frac{n}{2},\epsilon\tau\right)$. Applying the Chernoff bound (\ref{eq:chernoff}) to~$Z_{n/2}$, we have $\Prob{Z_{n/2} < -\epsilon\sqrt{n/2}} \leq e^{-\epsilon^2/2}$. It follows that for all
  \[
    \Prob{M_n < \mathrm{depth}\left(\frac{n}{2}, \epsilon\tau\right) - \epsilon\sqrt{n/2}}
    \leq 2\cdot \mathrm{prob}\left(\frac{n}{2},\epsilon\tau\right)+e^{-\epsilon^2/2}\;,
  \]
  i.e., by the definitions we have $\Prob{M_n < \mathrm{depth}(n,\epsilon)} \leq \mathrm{prob}(n,\epsilon)$. What remains is to find closed expressions for bounds on $\mathrm{depth}(n,\epsilon)$ and $\mathrm{prob}(n,\epsilon)$. For some choices of $\epsilon$, we might have a vacuous probability bound $\mathrm{prob}(n,\epsilon) \geq 1$, so the goal is to find $\epsilon>0$ with a useful bound.

  By unfolding the recursive definitions, and using $1 < \tau < \sqrt{2}$, we get
  \[
    \mathrm{depth}(n,\epsilon) 
    \geq -1 - \sum_{i=0}^{\infty} \epsilon\tau^i\sqrt{n/2^{i+1}}
    = -1 - \epsilon\sqrt{n/2}\sum_{i=0}^{\infty} \left(\frac{\tau}{\sqrt{2}}\right)^{i}
    = -1 - \frac{\epsilon}{\sqrt{2}-\tau}\sqrt{n}\;.
  \]
  For the probability bound we assume $\tau=1.3$, implying $(\tau^2)^i \geq \frac{5}{4}i$ for $i\geq 0$, and assume $\epsilon \geq \sqrt{8\ln 2}$. We get
  \[
    \mathrm{prob}(n,\epsilon) 
    \leq \sum_{i=0}^{\infty} 2^i\cdot e^{-(\epsilon\tau^i)^2/2}
    =  e^{-\epsilon^2/2} + \sum_{i=1}^{\infty} e^{i\ln 2-(\tau^2)^i\cdot \epsilon^2/2}
    \leq  e^{-\epsilon^2/2} + \sum_{i=1}^{\infty} e^{i\ln 2- \frac{5}{4} i\cdot\epsilon^2/2}
  \]
  \[
    \leq e^{-\epsilon^2/2} + \sum_{i=1}^{\infty} e^{-i \cdot \epsilon^2/2 }
    \leq e^{-\epsilon^2/2} + \frac{e^{-\epsilon^2/2}}{1-e^{-\epsilon^2/2}}
    \leq e^{-\epsilon^2/2} + \frac{e^{-\epsilon^2/2}}{1-\frac{1}{16}}
    \leq 3 \cdot e^{-\epsilon^2/2} \;.
  \]
  We conclude $\Prob{M_n < -1 - \frac{\epsilon}{\sqrt{2}-\tau}\sqrt{n}} \leq 3 \cdot e^{-\epsilon^2/2}$. It follows that $\E{M_n}= -O\left(\sqrt{n} \right)$, and  $\E{Y_n}=O\big(\sqrt{n}\big)$.
\end{proof}

Finally, we prove that the expected value of each $Y_n$ is bounded by a constant if $p^- > p^+$.

\begin{lemma}
  \label{lem:exp-decrease}
  For $p^- > p^+$, we have $\E{Y_n} \leq \frac{\alpha(1-\alpha)}{(1-2\alpha)^2}$, where $\alpha = \frac{p^+}{p^+ + p^-}<\frac{1}{2}$.
\end{lemma}

\begin{proof}
  Let $\pi_n^i=\Prob{Y_n = i}$. By the definition of $Y_n$, we have
  \[
    \pi_n^i = \left\{
      \begin{array}{cl}
        1 & \mbox{if } n = 0 \mbox{ and } i = 0 \\
        0 & \mbox{if } n = 0 \mbox{ and } i > 0 \\
        \pi_{n-1}^0 \cdot p_0 + \pi_{n-1}^1 \cdot p^- & \mbox{if } n > 0  \mbox{ and } i = 0 \\
        \pi_{n-1}^{i-1}  \cdot  p^+ + \pi_{n-1}^i \cdot p_0 + \pi_{n-1}^{i+1} \cdot  p^-  & \mbox{if } n > 0 \mbox{ and } i > 0\;.
      \end{array}
    \right.
  \]
  By induction in $n$, we prove $\pi_n^i \leq c^i$, for some constant $0 < c < 1$. For $n=0$, the inequality $\pi_n^i \leq c^i$ holds trivially, since $\pi_0^0 = 1 = c^0$ and $\pi_0^i=0<c^i$, for $i>0$ and $c>0$. For $n>0$, from the above last two cases we have the following condition (where the first inequality is from the induction hypothesis) 
  \[
    \pi_n^i \leq c^{i-1} \cdot p^+ + c^{i} \cdot p_0 + c^{i+1} \cdot p^- \leq c^{i}\;.
  \]
  Simplifying and using that $1-p_0=p^+ + p^-$, we get
  \[
    c^{i+1} \cdot p^- + c^i \cdot (p_0 - 1) + c^{i-1} \cdot p^+ \leq 0\;,
  \]
  \[
    c^2 \cdot p^- - c \cdot (p^+ + p^-) + p^+ \leq 0\;.
  \]
  Multiplying by $\frac{1}{p^+ + p^-}$, we have
  \[
    c^2 \cdot \frac{p^-}{p^+ + p^-} - c + \frac{p^+}{p^+ + p^-} \leq 0\;.
  \]
  Setting $\alpha = \frac{p^+}{p^+ + p^-}<\frac{1}{2}$, we get the final condition
  \[
    c^2 \cdot (1-\alpha) - c + \alpha \leq 0\;,
  \]
  that holds for $\frac{\alpha}{1-\alpha}\leq c \leq 1$. It follows that $\pi_n^i\leq \big(\frac{\alpha}{1-\alpha}\big)^i$ and $\E{Y_n} = \sum_{i=0}^{\infty} i \cdot \pi^i_n \leq  \sum_{i=0}^{\infty} i \cdot \big(\frac{\alpha}{1-\alpha}\big)^i = \frac{\alpha(1-\alpha)}{(1-2\alpha)^2}$ for all~$n$.
\end{proof}

A simulation of the process for various values of $p^-$ is shown in Figure~\ref{fig:experiments-comparison}. For each $p^+$ probability, we simulated $10^3$ counters for $10^6$ steps, and plotted the average of all counters. On the $y$-value is plotted $\lg \mathrm{value} / \lg n$, that converges to $c$ if value is $\Theta(n^c)$. Note how sensitive the curves are to $p^+$ around $p^+=\frac{1}{2}$.

\newpage

\section{Plots}

\begin{figure}[h]
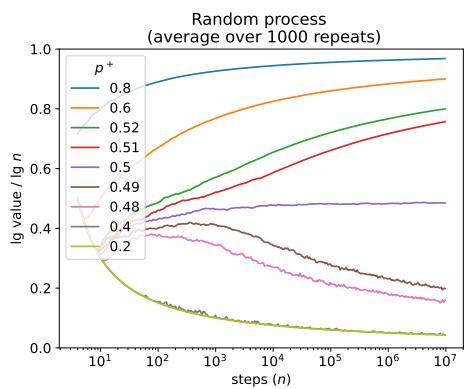

  \centering
  \plot{random-process}  
  \caption{Simulation of the random process studied in Theorem~\ref{thm:process} for various values of $p^+$ and $p^-=1-p^+$,  and $p_0=0$. If the expected value of the process is $\Theta(n^c)$, we would expect to see $\lg \mathrm{value} / \lg n$ converge to $c$.}
  \label{fig:random-process}
\end{figure}

\begin{figure}[h]
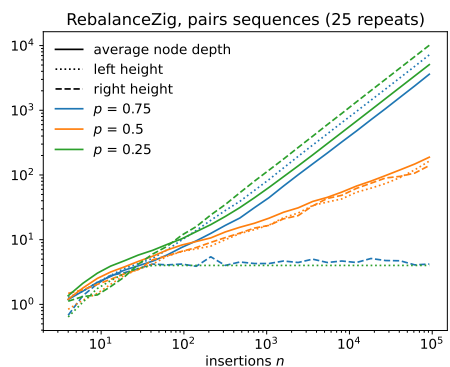

  \centering
  \plot{zig-pairs}
  \caption{Average node depths, left height, and right height for inserting pairs sequences using \RebalanceZig. The plotted data is the average over 25 runs.}
  \label{fig:zig-pairs}
\end{figure}
}

\begin{figure}[ht]
  \centering
  \plot[a]{zig-n1024}
  \plot[b]{zig-p0.5} \\
  \plot[c]{zigzag-n1024}
  \plot[d]{zigzag-p0.5}
\LongVersion{%
  \\
  \plot[e]{zigzig-n1024}
  \plot[f]{zigzig-p0.5}
}
  \caption{(left) The average node depths in binary search trees created by \RebalanceZig%
    \ShortVersion{ and \RebalanceZigZag}%
    \LongVersion{, \RebalanceZigZag and \RebalanceZigZig}%
    , respectively, for various types of choices of $p$ for insertion sequences of length~1024; (right) similar results but for fixed coin probability $p=\frac{1}{2}$, and sequence lengths in the range 1 to 1024. For $p=0$ all rotations are performed at the inserted node (and \RebalanceZig always creates a path), whereas no rotations are performed when $p=1$, i.e., unbalanced binary search trees.}
  \label{fig:experiments-comparison}
\end{figure}

\LongVersion{%
\begin{figure}
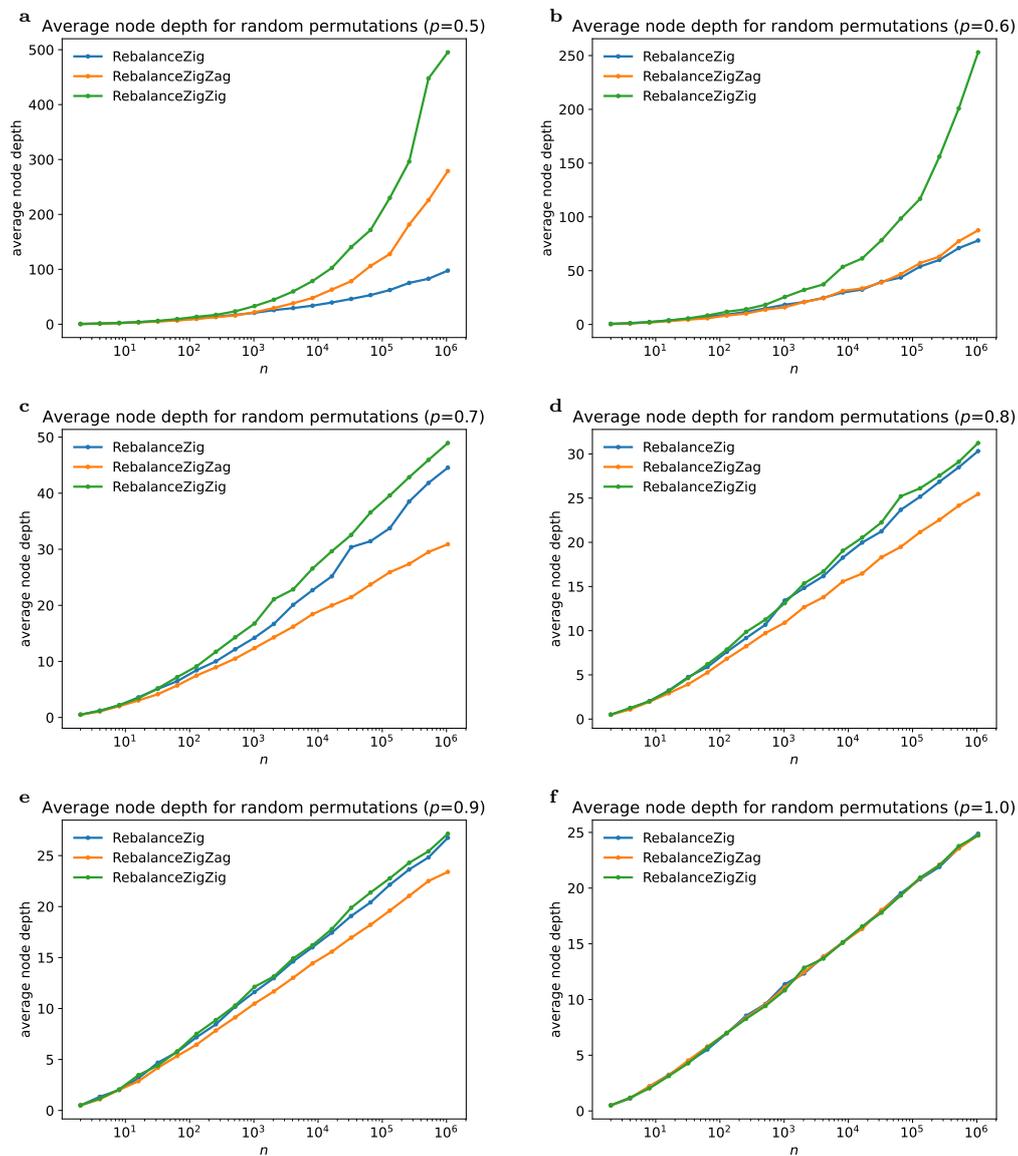

  \centering
  \plot[a]{permutation-p0.5}
  \plot[b]{permutation-p0.6} \\
  \plot[c]{permutation-p0.7}
  \plot[d]{permutation-p0.8} \\
  \plot[e]{permutation-p0.9}
  \plot[f]{permutation-p1.0}
  \caption{Average node depth for inserting random permutations with the three algorithms for different choices of $\frac{1}{2}\leq p\leq 1$.}
  \label{fig:experiments-permutations}
\end{figure}

\begin{figure}
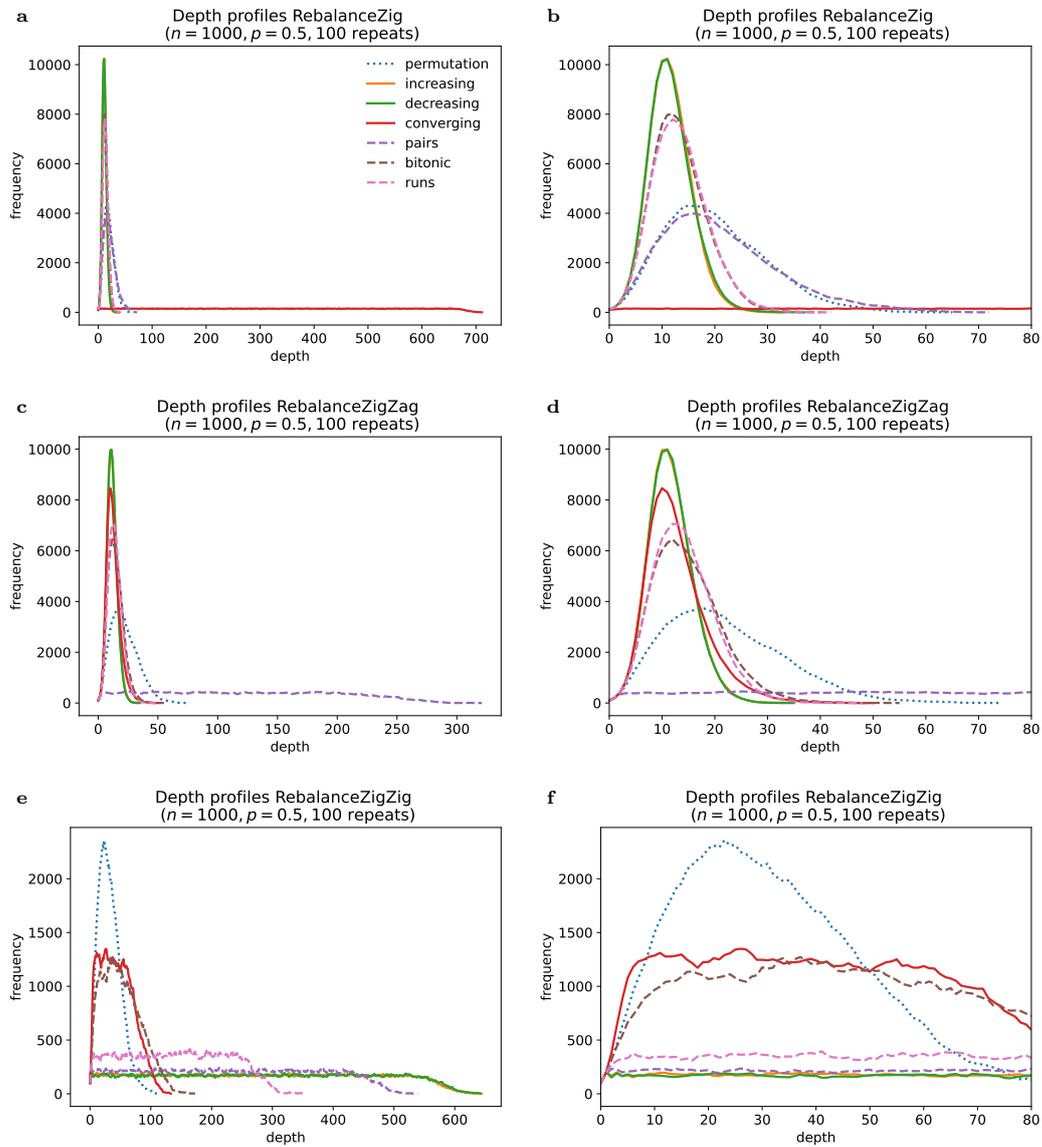

  \centering
  \plot[a]{profiles-zig}
  \plot[b]{profiles-zig-zoom} \\
  \plot[c]{profiles-zigzag}
  \plot[d]{profiles-zigzag-zoom} \\
  \plot[e]{profiles-zigzig}
  \plot[f]{profiles-zigzig-zoom} \\
  \caption{Depth profiles of the trees generated by the algorithms for different types of insertion sequences ($n=1000$, $p=\frac{1}{2}$, sum over 100 trees). Unbalanced is the result of inserting random permutations without rebalancing. (left) Full depth range; (right) zoomed in.}
  \label{fig:depth-profiles}
\end{figure}


\begin{figure}
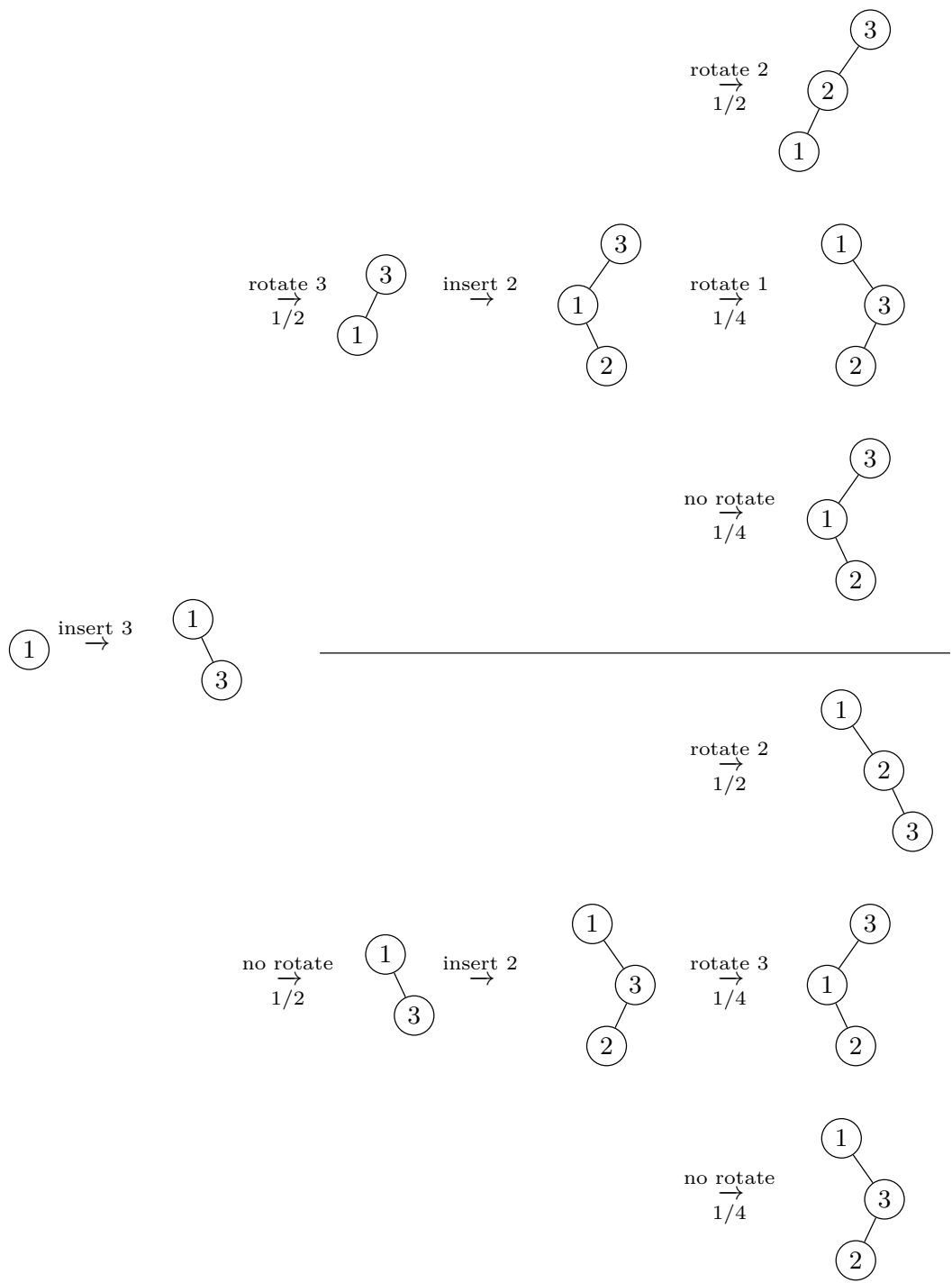

  \resizebox{\textwidth}{!}{%
  \begin{tabular}{ccccccccc}
    & & & & & & & \ARROW{rotate 2}{1/2} & \tree{[.3 [.2 [.1 ] \edge[missing]; \node[missing] {}; ] \edge[missing]; \node[missing] {}; ]} \\[10ex]
    & & & \ARROW{rotate 3}{1/2} & \tree{[.3 [.1 ] \edge[missing]; \node[missing] {}; ]}
    & \ARROW{insert 2}{} & \tree{[.3 [.1 \edge[missing]; \node[missing] {}; [.2 ] ] \edge[missing]; \node[missing] {}; ]}
    & \ARROW{rotate 1}{1/4} & \tree{[.1 \edge[missing]; \node[missing] {}; [.3 [.2 ] \edge[missing]; \node[missing] {}; ] ]} \\[10ex]
    & & & & & & & \ARROW{no rotate}{1/4} & \tree{[.3 [.1 \edge[missing]; \node[missing] {}; [.2 ] ] \edge[missing]; \node[missing] {}; ]}
    \\[10ex]
    \cline{5-9} \\
    \raisebox{11ex}[0ex]{\tree{[.1 ]}}
    & \raisebox{11ex}[0ex]{\ARROW{insert 3}{}} & \raisebox{11ex}[0ex]{\tree{[.1 \edge[missing]; \node[missing] {}; [.3 ] ]}} 
    & & & & & \ARROW{rotate 2}{1/2} & \tree{[.1 \edge[missing]; \node[missing] {}; [.2 \edge[missing]; \node[missing] {}; [.3 ] ] ]} \\[10ex]
    & & & \ARROW{no rotate}{1/2} & \tree{[.1 \edge[missing]; \node[missing] {}; [.3 ] ]} & \ARROW{insert 2}{} & \tree{[.1 \edge[missing]; \node[missing] {}; [.3 [.2 ] \edge[missing]; \node[missing] {}; ] ]}
    & \ARROW{rotate 3}{1/4} & \tree{[.3 [.1 \edge[missing]; \node[missing] {}; [.2 ] ] \edge[missing]; \node[missing] {}; ]} \\[10ex]
    & & & & & & & \ARROW{no rotate}{1/4} & \tree{[.1 \edge[missing]; \node[missing] {}; [.3 [.2 ] \edge[missing]; \node[missing] {}; ] ]}   \end{tabular}}
  \caption{Binary search trees resulting from inserting the sequence 1, 3, 2 using \RebalanceZig with $p=\frac{1}{2}$. Below each arrow is the probability that this choice is taken.}
  \label{fig:insertion-cases}
\end{figure}
}

\end{document}